\documentclass[twoside,a4paper]{article}
\usepackage{actes}

\usepackage{amsfonts}
\usepackage{amsmath}
\usepackage{amssymb}
\usepackage{amsthm}
\usepackage[english]{babel}
\let\textcite\cite  
\usepackage{caption}
\usepackage{color}
\usepackage{enumerate}
\usepackage[hidelinks]{hyperref}
\usepackage{hyperxmp}
\usepackage{mathtools}
\usepackage{paralist}
\usepackage{stmaryrd}
\usepackage{xcolor,verbatim}

\usepackage[T1]{fontenc}
\usepackage[utf8]{inputenc}

\definecolor{rgreen}{HTML}{004400}
\definecolor{rblue}{HTML}{000044}

\newif\ifcolourversion
\colourversiontrue
\colourversionfalse
\newif\ifarxivversion
\arxivversionfalse
\arxivversiontrue

\ifcolourversion%
\usepackage[left]{showlabels}%
\newcommand{\arxiv}[2]{{\color{rgreen}#1\color{rblue}#2}}%
\else%
\newcommand{\arxiv}[2]{%
  \ifarxivversion%
  #1%
  \else%
  #2%
  \fi}
\fi

\ifarxivversion
\hypersetup{%
  pdftitle={On Coupled Logical Bisimulation},
  pdfauthor={Ryan Kavanagh <ryan@cs.queensu.ca>},
  pdfkeywords={bisimulation} {up-to techniques},
  pdfcopyright={Copyright (C) 2014 Ryan Kavanagh. All rights
    reserved.},
  pdflang={English},
}
\fi

\ifarxivversion%
\usepackage{microtype}
\IfFileExists{MinionPro.sty}{%
  \usepackage{MinionPro}
  
}{%
  \usepackage{mathpazo,eulervm}
}
\fi

\DeclareFontFamily{U}{mathx}{\hyphenchar\font45}
\DeclareFontShape{U}{mathx}{m}{n}{
  <5> <6> <7> <8> <9> <10>
  <10.95> <12> <14.4> <17.28> <20.74> <24.88>
  mathx10
}{}
\DeclareSymbolFont{mathx}{U}{mathx}{m}{n}
\DeclareFontSubstitution{U}{mathx}{m}{n}
\DeclareMathAccent{\widecheck}{0}{mathx}{"71}

\newcounter{DHsupercounter}
\newtheorem{lemma}[DHsupercounter]{Lemma}
\newtheorem{theorem}[DHsupercounter]{Theorem}
\newtheorem{corollary}[DHsupercounter]{Corollary}
\newtheorem{prop}[DHsupercounter]{Proposition}

\theoremstyle{definition}

\newtheorem{definition}[DHsupercounter]{Definition}

\newcommand{\ra}{\longrightarrow}
\newcommand{\nra}{{\not\ra}}
\newcommand{\Ra}{\Longrightarrow}
\newcommand{\R}{\mathrel{\mathcal{R}}}
\newcommand{\SR}{\mathrel{\mathcal{S}}}
\newcommand{\Ld}{\Lambda^\bullet}
\newcommand{\clb}{\approx}
\newcommand{\clbv}{\clb^v}
\newcommand{\clbn}{\clb^n}
\newcommand{\wt}{\widetilde}

\newcommand{\EC}{\reflectbox{C}}
\newcommand{\conv}{{\Downarrow}}
\newcommand{\convr}{\Downarrow}
\newcommand{\diver}{{\Uparrow}}

\newcommand{\ece}{\approxeq}
\newcommand{\cen}{\simeq^n}
\newcommand{\ecen}{\approxeq^n}
\newcommand{\cev}{\simeq^v}
\newcommand{\ecev}{\approxeq^v}
\newcommand{\ABn}{\approx^n_A}
\newcommand{\ABv}{\approx^v_A}
\newcommand{\V}{\mathcal{V}}

\newcommand{\VC}{\mathfrak{V}}
\newcommand{\fl}[1]{\langle #1 \rangle}

\newcommand{\vr}[1]{\mathrel{{#1}|_{\V\times\V}}}
\newcommand{\prog}{\rightarrowtail}
\newcommand{\progl}{\rightharpoondown}

\newcommand{\F}{\mathcal{F}}
\newcommand{\utr}[1]{\Rightarrow#1\Leftarrow}
\newcommand{\eccnnp}[2]{\mathrel{#1\otimes_n#2}}
\newcommand{\eccn}[2]{\mathrel{({\eccnnp{#1}{#2}})}}
\newcommand{\eccvnp}[2]{\mathrel{#1\otimes_v#2}}
\newcommand{\eccv}[2]{\mathrel{({\eccvnp{#1}{#2}})}}
\newcommand{\lbn}{\approx^{l\!n}}
\newcommand{\lbv}{\approx^{l\!v}}
\newcommand{\Id}{\mathrm{Id}}

\newcommand{\op}{^{op}}

\newcommand{\oclos}{\circ}
\newcommand{\bull}{\spadesuit}
\renewcommand{\P}{\mathcal{P}}
\newcommand{\Pev}{\P_{ev}}
\newcommand{\Q}{\mathcal{Q}}
\newcommand{\U}{\mathcal{U}}
\newcommand{\N}{\mathbb{N}}

\newcommand{\pfun}{\rightharpoonup}

\DeclareRobustCommand\longtwoheadrightarrow
     {\relbar\joinrel\twoheadrightarrow}

\usepackage{todonotes}
\newcommand{\DHC}[1]{{\color{orange} #1}}
\newcommand{\RAK}[1]{{\color{red} #1}}
\renewcommand{\RAK}[1]{} 
\newcommand{\JMM}[1]{{\color[rgb]{0,.7,0} #1}}
\renewcommand{\JMM}[1]{\doesnotcompile}

\newcommand\restr[2]{{
    \left.\kern-\nulldelimiterspace 
      #1 
      \vphantom{\big|} 
    \right|_{#2} 
  }}

\DeclareMathOperator{\fv}{fv}
\DeclareMathOperator{\bv}{bv}

\title{On Coupled Logical Bisimulation for the $\lambda$-Calculus}
\author{Ryan Kavanagh$^{1,2}$ \& Jean-Marie Madiot$^2$
  \\ {\tt ryan@cs.queensu.ca}, {\tt jeanmarie.madiot@ens-lyon.fr}
}
\titlehead{On Coupled Logical Bisimulation for the $\lambda$-Calculus}
\authorhead{Kavanagh \& Madiot}
\affiliation{\begin{tabular}{rr}
\\ 1:  School of Computing
\\     Queen's University at Kingston
\\     Kingston, Ontario, Canada
\\ 2:  Laboratoire de l'Informatique du Parallélisme
\\     École normale supérieure de Lyon
\\     Lyon, France
\end{tabular}}

\begin{document}
\setcounter{page}{1}
\maketitle

\begin{abstract}
  We study coupled logical bisimulation (CLB) to reason
  about contextual equivalence in the $\lambda$-calculus. CLB
  originates in a work by Dal Lago, Sangiorgi and Alberti, as a tool
  to reason about a $\lambda$-calculus with probabilistic
  constructs. We adapt the original definition to the pure
  $\lambda$-calculus. 
  We develop the metatheory of CLB in call-by-name and in
  call-by-value, and draw comparisons with applicative bisimulation
  (due to Abramsky) and logical bisimulation (due to Sangiorgi,
  Kobayashi and Sumii).
  We also study enhancements of the bisimulation method for CLB by
  developing a theory of up-to techniques for cases where the
  functional corresponding to bisimulation is not necessarily
  monotone.
\end{abstract}

\section{Introduction}

Several coinductive methods to reason about equivalences between
higher-order programs or processes have been proposed. The starting
point in this direction is Abramsky's Applicative Bisimulation
(AB)~\cite{Abramsky1990}.
Several alternatives to AB have been proposed since, with two main
objectives. A first objective is to be able to develop the metatheory
of the bisimilarity in a simple way.
The main question related to this objective is to prove that
bisimilarity coincides with contextual equivalence.
A second objective is to be able to equip the coinductive method with
powerful proof techniques, that allow one to avoid including redundant
pairs in the relations being studied.
These so-called \emph{up-to techniques}~\cite{Sangiorgi2012a} can turn out
to be very useful in developing proofs of equivalence between
programs.

To address these objectives, Logical Bisimulation
(LB)~\cite{Sangiorgi2007} and, successively, Environmental
Bisimulation (EB)~\cite{Sangiorgi2011b} have been proposed.
LB and EB can both be seen as improvements of AB.

\medskip

Recently, in a study of a $\lambda$-calculus enriched with
probabilistic features~\cite{DalLago2014}, Dal Lago, Sangiorgi and Alberti have
introduced Coupled Logical Bisimulation (CLB). In that work, CLB is
motivated by technical considerations, related to the way
probabilistic $\lambda$-terms evolve and can be observed.
The main purpose of the present work is to understand how CLB compares
to existing notions of bisimulation for higher-order calculi. To
achieve this, we formulate CLB in the simpler setting of the pure
$\lambda$-calculus, and study its basic metatheory, as well as some
up-to techniques.


We consider CLB in both the call-by-name and the call-by-value
$\lambda$-calculus, relating it with AB and LB.
To define bisimulation enhancements for CLB, the existing theory of
bisimulation enhancements~\cite{Sangiorgi2012a} cannot be reused
directly. The reason is that, like in the case of LB, the functional
associated to the definition of bisimulation is not monotone.
We show how sound up-to techniques for bisimulation can be adapted in
this setting, and how to compose them.

Another contribution of the paper is a written proof of the
\emph{context lemma} for the call-by-value $\lambda$-calculus, which
says that contextual equivalence is preserved when restricting to
evaluation contexts only.
While this result seems to belong to folklore, we have not been able
to find it in the literature.
The proof is not a direct adaptation of the corresponding result in
call-by-name~\cite{Milner1977}.


\paragraph{Outline of the paper.}
In Section~\ref{sec:genupto}, we present a theory of up-to
techniques in absence of monotonicity of the functional corresponding
to bisimulation. Section~\ref{s:prelim} recalls some general notions
about the $\lambda$-calculus. 
We study CLB for the call-by-name $\lambda$-calculus in
  Section~\ref{s:cbn}.
We then present the main properties of CLB for the
call-by-value $\lambda$-calculus in Section~\ref{s:cbv}.
In passing, we present in Section~\ref{s:context:lemma:cbv} a proof of
the ``context lemma'' for the call-by-value $\lambda$-calculus.

\section{
Up-to techniques in absence of monotonicity
}
\label{sec:genupto}

As we shall see in Section~\ref{s:cbn}, the functional defining CLB is
not monotone.
In this section, we present an axiomatic theory of up-to techniques
which does not rely on this property.

Although our theory is motivated by up-to techniques for bisimulation
relations, it applies equally well to any theory comprising a family
$\F$ of ``accepted relations'' contained in a universe $\U$ of all
possible ``relations''.
Although our theory is inspired from bisimulations, we deliberately
avoid specifying the nature of the ``relations'' in $\U$: these need
not be relations in the traditional sense of sets of pairs of terms.
For example, our theory applies equally well to coupled logical
bisimulations, presented below, which are pairs of relations, and it
could be used to reason about any desired subset $\F$ of some universe
$\U$.

\begin{definition}
  Given a family $\F \subseteq \U$ of accepted relations, a
  \textbf{progression} for $\F$ is a relation ${\prog} \subseteq \U
  \times \U$ such that ${\R} \prog {\R}$ only if there exists an
  ${\R'} \in \F$ such that ${\R} \subseteq {\R'}$. If ${\R} \prog
  {\SR}$, then we say that $\R$ \textbf{progresses to} $\SR$.
\end{definition}

Although up-to techniques for coinductively-defined families of
relations are often given as total functions over $\U$, we relax this
definition to the following:

\begin{definition}
  \label{def:fuptot}
  Given a family of relations $\F \subseteq \U$ with a progression, we
  call a partial function $\P : \U \pfun \U$ an \textbf{up-to
    technique} and say that it is \textbf{sound} if whenever
  $\P({\R})$ is defined and ${\R} \prog \P({\R})$, there exists an
  ${\R'} \in \F$ such that ${\R} \subseteq {\R'}$.
  We say that $\P$ is \textbf{monotone} if whenever ${\R} \subseteq
  {\SR}$ and both $\P({\R})$ and $\P({\SR})$ are defined, then
  $\P({\R}) \subseteq \P({\SR})$.
\end{definition}

\begin{definition}\label{def:finiteconverge}
  We say that an up-to technique $\P$ is \textbf{finitely convergent}
  if there exists an $N$ such that for all $n, m > N$, $\P^n = \P^m$;
  in such situation, we call $N$ the \textbf{finite convergence constant}.
  Finally, we say that two up-to techniques $\P$ and $\Q$
  \textbf{commute} if $\P\circ\Q = \Q\circ\P$.
\end{definition}

We also make use of the following definitions, which are based on those in~\textcite{Pous2012}:

\begin{definition}
  We say that an up-to technique $\P$ is \textbf{compatible} if for
  all $\R$ and $\SR$ such that ${\R} \prog {\SR}$, $\P({\R})$ and
  $\P({\SR})$ are defined and $\P({\R}) \prog \P({\SR})$.
  We say that it is \textbf{respectfully compatible} if for all $\R$
  and $\SR$ such that ${\R} \subseteq {\SR}$ and ${\R} \prog {\SR}$,
  $\P({\R})$ and $\P({\SR})$ are defined, and $\P({\R}) \subseteq
  \P({\SR})$ and $\P({\R}) \prog \P({\SR})$ hold.
  We say that an up-to technique $\P$ is \textbf{extensive} if for all
  $\R$, if $\P({\R})$ is defined, then ${\R} \subseteq \P({\R})$.
\end{definition}

We remark that, although the definitions of ``monotone, compatible
up-to technique'' and ``respectfully compatible up-to technique'' are
similar, monotony is a much stronger condition than respectfulness.
We also observe that respectfully compatible up-to techniques need not
be compatible and vice-versa.

From these straightforward definitions, we can deduce sufficient
conditions for the soundness of up-to techniques.
As we shall see below, imposing a condition known as ``continuity'' on
our progression relation provides significantly simpler sufficiency
conditions for soundness.

\begin{prop}
  \label{prop:fcecsound}
  A 
  finitely convergent, extensive, (respectfully) compatible
  up-to technique $\P$ is sound.
\end{prop}

\begin{proof}
  If $\P({\R})$ is defined and ${\R} \prog \P({\R})$, we prove by
  induction on $n$ that for all $n$, $\P^{n}({\R})$ is defined,
  $\P^{n}({\R}) \prog \P^{n+1}({\R})$ and ${\R} \subseteq
  \P^{n}({\R})$.  Hence, $\P^{N+1}({\R}) \prog \P^{N+2}({\R}) =
  \P^{N+1}({\R})$ by finite convergence. Thus for some ${\R'}$,
  $\P^{N+1}({\R}) \subseteq {\R'} \in \F$, which implies ${\R}
  \subseteq {\R'} \in \F$.
\end{proof}


\arxiv{As a corollary of the proof, we get that:}{}

\begin{corollary}
  \label{cor:fcecsound}
  If $\P$ is finitely convergent, extensive, and (respectfully)
  compatible and ${\R} \prog \P({\R})$, then, where $N$ is the finite
  convergence constant, for all $n > N$ we have $\P^{n}({\R}) \prog
  \P^{n}({\R})$ and $\P^n({\R}) \subseteq {\R'}$ for some ${\R'} \in
  \F$.
\end{corollary}

\arxiv{
  \begin{proof}
    Repeatedly applying compatibility to $\P^{N+1}({\R}) \prog
    \P^{N+1}({\R})$, we get that $\P^{n}({\R}) \prog \P^{n}({\R})$ for
    all $n > N$, and so $\P^{n}({\R}) \subseteq {\R'} \in \F$ for all
    $n > N$.
  \end{proof}
}{}

\arxiv{
The following propositions tell us that composition of up-to
techniques is well-behaved. They follow straightforwardly from the
corresponding definitions.

\begin{prop}
  \label{prop:extensivecomp}
  If $\P$ and $\Q$ are two extensive up-to techniques, then so is
  $\Q\circ\P$.
\end{prop}

\arxiv{
  \begin{proof}
    Assume $(\Q\circ\P)({\R})$ is defined, then the proposition is
    immediate by transitivity: ${\R} \subseteq \P({\R}) \subseteq
    \Q(\P({\R})) = (\Q\circ\P)({\R})$.
  \end{proof}
}{}

\begin{prop}
  \label{prop:compatiblecomp}
  If $\P$ and $\Q$ are two (respectfully) compatible up-to techniques,
  then so is $\Q\circ\P$.
\end{prop}

\arxiv{
  \begin{proof}
    We show that compatibility is composed by composition; respectful
    compatibility follows in an identical manner.
    By compatibility of $\P$, for all ${\R} \prog {\SR}$, $\P({\R})$
    and $\P({\SR})$ are defined and $\P({\R}) \prog \P({\SR})$.
    By compatibility of $\Q$, for all ${\R} \prog {\SR}$, $\Q({\R})$
    and $\Q({\SR})$ are defined and $\Q({\R}) \prog \Q({\SR})$.
    Combining these two facts, we get that for all ${\R} \prog {\SR}$,
    $(\Q\circ\P)({\R})$ and $(\Q\circ\P)({\SR})$ are defined and
    $(\Q\circ\P)({\R}) \prog (\Q\circ\P)({\SR})$, i.e., $\Q\circ\P$ is
    compatible.
\end{proof}
}{}

\begin{prop}
  \label{prop:convergentcomp}
  If $\P$ and $\Q$ are two finitely convergent up-to techniques that
  commute, then $\Q\circ\P$ is finitely convergent.
\end{prop}

\arxiv{
  \begin{proof}
    Let $M$ and $N$ be the convergence constants of $\P$ and $\Q$
    respectively, and let $L = \max(M, N)$.
    By commutativity, $(\Q\circ\P)^k = \Q^k\circ\P^k$ for all $k$, and
    so for all $m, n > L$, $(\Q\circ\P)^m = \Q^m\circ\P^m =
    \Q^n\circ\P^n = (\Q\circ\P)^n$.
    Thus, $\Q\circ\P$ is finitely convergent.
  \end{proof}
}{}

\begin{corollary}
  If $\P$ and $\Q$ are two extensive, (respectfully) compatible, and
  finitely convergent up-to techniques that commute, then $\Q\circ\P$
  is sound.
\end{corollary}

\arxiv{
  \begin{proof}
    The composition $\Q\circ\P$ satisfies the hypotheses of
    Proposition \ref{prop:fcecsound} by Propositions
    \ref{prop:extensivecomp}, \ref{prop:compatiblecomp}, and
    \ref{prop:convergentcomp}.
  \end{proof}
}{} } { The following proposition establishes that the composition of
up-to techniques is well-behaved. The results follow straightforwardly
from the corresponding definitions.
\begin{prop}\label{props:composition} Suppose $\P$ and $\Q$ are up-to techniques.
  Then we have:
  \begin{itemize}
  \item If $\P$ and $\Q$ are extensive, then so is $\Q\circ\P$.
  \item   If $\P$ and $\Q$ are (respectfully) compatible,
  then so is $\Q\circ\P$.
\item   If $\P$ and $\Q$ are finitely convergent and
  commute, then $\Q\circ\P$ is finitely convergent.
\item If $\P$ and $\Q$ are extensive, (respectfully) compatible,
   finitely convergent, and commute, then
   $\Q\circ\P$ is sound.
  \end{itemize}
\end{prop}
}

When dealing with monotone up-to techniques, we can relax the
commutativity requirement at the expense of additional hypotheses.

\begin{lemma}
  \label{lemma:monotuptocommute}
  If $\P$ and $\Q$ are two monotone functions such that
  $(\Q\circ\P)({\R}) \subseteq (\P\circ\Q)({\R})$ for all $\R$, then
  for all $\R$ and all $k$, $(\Q\circ\P)^k({\R}) \subseteq
  (\P^k\circ\Q^k)({\R})$.
\end{lemma}

\arxiv{
\begin{proof}
  We proceed by induction on $k$.
  The case $k = 1$ is by hypothesis, so assume $(\Q\circ\P)^k({\R})
  \subseteq (\P^k\circ\Q^k)({\R})$ for some $k$.
  Clearly
  \begin{equation}
    \label{eq:monotuptocommute1}
    (\Q\circ\P)^{k+1}({\R}) = \left(\left(\Q\circ\P\right)
      \circ \left(\Q\circ\P\right)^k \right)({\R}),
  \end{equation}
  and by hypothesis,
  \begin{equation}
    \label{eq:monotuptocommute2}
  \left(\left(\Q\circ\P\right) \circ \left(\Q\circ\P\right)^k
  \right)({\R}) \subseteq \left(\left(\P\circ\Q\right) \circ
    \left(\Q\circ\P\right)^k \right)({\R}).
   \end{equation}
  Since the composition of
  monotone functions is monotone, we get by the induction hypothesis
  that
  \begin{align*}
    \left(\left(\P\circ\Q\right) \circ \left(\Q\circ\P\right)^k
    \right)({\R}) &\subseteq \left(\left(\P\circ\Q\right) \circ
      \left(\P^k\circ\Q^k\right) \right)({\R})\\
    &= \left(\P\circ\left(\Q\circ \P \right) \circ
      \left(\P^{k-1}\circ\Q^k\right) \right)({\R})
  \end{align*}
  However, using once again the hypothesis that $(\Q\circ\P)({\R})
  \subseteq (\P\circ\Q)({\R})$ for all $\R$, we get
  that \[\left(\left(\Q\circ P \right) \circ
    \left(\P^{k-1}\circ\Q^k\right) \right)({\R}) \subseteq
  \left(\left(\P\circ Q \right) \circ \left(\P^{k-1}\circ\Q^k\right)
  \right)({\R}),\] and so by the monotony of $\P$, we get \[
  \left(\P\circ\left(\Q\circ \P \right) \circ
    \left(\P^{k-1}\circ\Q^k\right) \right)({\R}) \subseteq
  \left(\P\circ\left(\P\circ \Q \right) \circ
    \left(\P^{k-1}\circ\Q^k\right) \right)({\R}).\] Repeating in this
  manner, we get that
  \begin{equation}
    \label{eq:monotuptocommute3}
    \left(\P^i\circ\Q\circ\P^{k+1-i} Q^k \right) \subseteq
    \left(\P^{i+1}\circ\Q\circ\P^{k-i} Q^k \right)
  \end{equation}
  for all $0 \leq i \leq k$.
  Thus, by transitivity using the inclusions
  \eqref{eq:monotuptocommute1}, \eqref{eq:monotuptocommute2}, and
  \eqref{eq:monotuptocommute3} all the way up to $i = k$, we get that
  $(\Q\circ\P)^{k+1}({\R}) \subseteq (\P^{k+1}\circ\Q^{k+1})({\R})$ as
  desired.
  We thus conclude the lemma by induction.
\end{proof}
}{}

\begin{prop}
  If $\P$ and $\Q$ are two monotone, finitely converging,
  (respectfully) compatible, and extensive up-to techniques such that
  $(\Q\circ\P)({\R}) \subseteq (\P\circ\Q)({\R})$ for all $\R$, and
  ${\R} \prog (\Q\circ\P)({\R})$ implies ${\R} \prog \Q({\R})$, then
  $\P\circ\Q$ is sound.
\end{prop}

\begin{proof}
  Assume ${\R} \prog (\Q\circ\P)({\R})$.
  Let $N$ and $M$ be the convergence constants of $\P$ and $\Q$
  respectively (cf.\ Definition~\ref{def:finiteconverge}), 
  and let $L = \max(N, M) + 1$.
  Then by Corollary \ref{cor:fcecsound}, we have that $\Q^L({\R})
  \prog \Q^L({\R})$, and so by compatibility of $\P$, we get that
  $(\P^L\circ\Q^L)({\R}) \prog (\P^L\circ\Q^L)({\R})$.
  Thus, we have that $(\P^L\circ\Q^L)({\R}) \subseteq {\R'}$ for some
  ${\R'} \in \F$.
  By
  \arxiv{Proposition~\ref{prop:extensivecomp},}
  {Proposition~\ref{props:composition},} $(\Q\circ\P)$ is extensive,
  so by repeated application of extensiveness and transitivity, we get
  that ${\R} \subseteq (\Q\circ\P)^L(\R)$.
  By Lemma \ref{lemma:monotuptocommute}, we get that
  $(\Q\circ\P)^L({\R}) \subseteq (\P^L\circ\Q^L)({\R})$.
  Thus, since we have ${\R} \subseteq (\P^L\circ\Q^L)({\R}) \subseteq
  {\R'} \in \F$ by transitivity, we conclude that $\P\circ\Q$ is
  sound.
\end{proof}

By using a stronger notion of progression, we can drop the
hypotheses of commutation and finite convergence
when showing soundness of compositions and up-to techniques
in general.
This stronger notion may seem \textit{ad hoc}, but we will see that it
is satisfied by the canonical progressions of well known
bisimulations.

\begin{definition}
  A progression $\prog$ is said to be \textbf{continuous} if for all
  ascending chains of relations ${\R_0} \subseteq {\R_1} \subseteq
  \cdots$ and ${\SR_0} \subseteq {\SR_1} \subseteq \cdots$ such that
  ${\R_i} \prog {\SR_i}$ for all $i$, we have $\nu{\R} \prog
  \nu{\SR}$, where $\nu{\R} = \bigcup_{i\in\N} {\R_i}$ and $\nu{\SR} =
  \bigcup_{i\in\N} {\SR_i}$.
\end{definition}

\arxiv{}{%
  Using extensiveness, it is easy to show by double inclusion that:
}

\begin{prop}
  \label{prop:extensivefp}
  If $\P$ is an extensive up-to technique, then for all $\R$, and
  where $\P^i$ is the $i$-th iterate of $\P$, $\nu\P({\R}) =
  \bigcup_{i\in\N} \P^i({\R})$ is a fixpoint for $\P$.
\end{prop}

\arxiv{
  \begin{proof}
    We want to show that $\P(\nu\P({\R})) = \nu\P({\R})$.
    The inclusion $\nu\P({\R}) \subseteq \P(\nu\P({\R}))$ is immediate
    by extensiveness, so assume ${\SR} \in \P(\nu\P({\R}))$.
    Then there exists an $i \in \N$ such that ${\SR} \in
    \P(\P^i({\R}))$.
    Then ${\SR} \in \P^{i+1}({\R})$, so ${\SR} \in \nu\P({\R})$.
    We conclude the equality by double inclusion.
  \end{proof}
}{}

\begin{prop}
  \label{prop:sympasound}
  If $\F$ has a continuous progression $\prog$ and $\P$ is a
  (respectfully) compatible, extensive up-to technique,
  then $\P$ is sound and $\nu\P({\R}) \prog \nu\P({\R})$ whenever
  ${\R} \prog \P({\R})$.
\end{prop}

\begin{proof}
  Assume ${\R} \prog \P({\R})$, then, by extensiveness, ${\R} \subseteq
  \P({\R})$.
  Then for all $k \in \N$, $\P^k({\R}) \prog \P^{k+1}({\R})$ by
  (respectful) compatibility, and $\P^k({\R}) \subseteq
  \P^{k+1}({\R})$ by extensiveness.
  Let ${\R_i} = \P^i({\R})$ and ${\SR_i} = \P^{i+1}({\R})$, then by
  continuity, $\nu{\R} \prog \nu{\SR}$.
  However, $\nu{\R} = \nu\P = \nu{\SR}$, so there exists an ${\R'} \in
  \F$ such that $\nu{\R} \subseteq {\R'}$.
  Thus, ${\R} \subseteq \nu{\R} \subseteq {\R'} \in \F$ and we
  conclude that $\P$ is sound.
\end{proof}

\begin{corollary}
  If $\F$ has a continuous progression and $\P$ and $\Q$ are two
  (respectfully) compatible, extensive up-to techniques, then
  $\Q\circ\P$ is sound.
\end{corollary}

\arxiv{
\begin{proof}
  The up-to technique $\Q\circ\P$ is extensive by Proposition
  \ref{prop:extensivecomp} and compatible by Proposition
  \ref{prop:compatiblecomp}, and so sound by Proposition
  \ref{prop:sympasound}.
\end{proof}
}{}

\begin{corollary}
  \label{cor:nupsound}
  If $\F$ has a continuous progression and $\P$ is a (respectfully)
  compatible, extensive up-to techniques, then $\nu\P$ is
  (respectfully) compatible, extensive, and sound.
\end{corollary}

Finally, it is often useful to consider up-to techniques which are in
a certain manner ``asymmetric''.
For example, for many types of bisimulation, soundness of
``weak bisimilarity up-to-bisimilarity'' requires us to use strong bisimilarity on the
side of the term making a small-step transition, and permits weak
bisimilarity on the side of the term answering with a large-step
transition; this discussion will be made more precise with examples in
the following sections.

\begin{definition}
  A \textbf{one-sided progression} is a relation ${\progl} \subseteq
  \U \times \U$ such that both ${\R} \progl {\R}$ and ${\R\op} \progl
  {\R\op}$ only if there exists ${\R'} \in \F$ such that ${\R}
  \subseteq {\R'}$.
\end{definition}

\begin{definition}
  \label{def:asfuptot}
  Given a family of relations $\F \subseteq \U$ with a one-sided
  progression, we call a partial function $\P : \U \pfun \U$ an
  \textbf{asymmetric up-to technique}, and say that it is
  \textbf{sound} if whenever $\P({\R})$ and $\P({\R\op})$ are defined
  and ${\R} \progl \P({\R})$ and ${\R\op} \progl \P({\R\op})$, there
  exists an ${\R'} \in \F$ such that ${\R} \subseteq {\R'}$.
\end{definition}

We adapt the definitions of commutativity, extensiveness, finite
convergence, and commutativity in the obvious manner, substituting
one-sided progression for progression.

\begin{prop}
  Every finitely convergent, extensive, (respectfully) compatible
  asymmetric up-to-technique $\P$ is sound.
\end{prop}

\arxiv{
\begin{proof}
  Adapt the proof of Proposition \ref{prop:fcecsound}, starting from
  ${\R} \progl \P({\R})$ and ${\R\op} \progl \P({\R\op})$, reaching
  $\P^{N+1}({\R}) \progl \P^{N+1}({\R})$ and $\P^{N+1}({\R\op}) \progl
  \P^{N+1}({\R\op})$, and deducing that $\P^{N+1}({\R})$ is contained
  in some ${\R'} \in \F$.
  Conclude that $\P$ is sound since by extensiveness and transitivity,
  ${\R} \subseteq {\R'} \in \F$ and $\R$ was arbitrary.
\end{proof}
}{}

\begin{prop}
  If $\F$ is has a continuous one-sided progression, and $\P$ is an
  extensive, (respectfully) compatible and asymmetric up-to technique,
  then $\P$ is sound.
\end{prop}

\arxiv{
\begin{proof}
  Adapt the proof of Proposition \ref{prop:sympasound}, starting from
  ${\R} \progl \P({\R})$ and ${\R\op} \progl \P({\R\op})$, reaching
  $\nu{\R} \progl \nu{\R}$ and $(\nu{\R})\op \progl (\nu{\R})\op$, and
  deducing that ${\R} \subseteq \nu{\R} \subseteq {\R'} \in \F$.
\end{proof}
}{}


\section{Preliminaries on the $\lambda$-calculus}
\label{s:prelim}

\arxiv{ We establish notation for the $\lambda$-calculus and prove a few
 useful lemmas about its contexts.}{}
We denote by $\Lambda$ the set of all $\lambda$ terms, and by $\Ld$
the set of all closed $\lambda$ terms, i.e., those with no free
variables.
If ${\R} \subseteq \Lambda\times\Lambda$ is a relation, and
$\wt{M}=(M_1, \dotsc, M_n)$ and $\wt{N} = (N_1,\dotsc,N_n)$ are
vectors in $\Lambda^n$, then we write $\wt{M}\R\wt{N}$ for $(M_1\R
N_1) \land \cdots \land (M_n\R N_n)$.
By abuse of notation, if $X \in \Lambda$, $\wt{M} \in \Lambda^m$, and
$\wt{N} \in \Lambda^n$, we write $\wt{M}X\wt{N}$ for the term $M_1\cdots M_m X
N_1\cdots N_n$.
Finally, if $\R$ is a relation and $S \subseteq \Lambda\times\Lambda$,
we denote by $\restr{{\R}}{S} := {\R} \cap S$ the restriction of $\R$ to
$S$.

\arxiv{
\begin{definition}
  If $\R,\R'\ \subseteq \Lambda\times\Lambda$ are two relations, then
  write $\R\R'$ for their \textbf{composition}, i.e., $M \R\R' N$ if
  and only if there exists an $L$ such that $M \R L \R' N$.
\end{definition}
}{}

We recall the familiar notion of context (see, e.g.,
\cite{Barendregt1985}) and Gordon's \cite{Gordon1994} canonical
contexts:

\begin{definition}
\arxiv{
  A \textbf{context} $C$ is given by the following grammar,
  \[ C := x \mid [\cdot] \mid C_1 C_2 \mid \lambda x.C.
  \] Let a \textbf{canonical context} be a context $\VC$ given by the
  grammar
  \[ \VC = [\cdot] \mid \lambda x.C, \]
  where $C$ ranges over all contexts.
}{
We define \textbf{contexts}, $C$, and  \textbf{canonical contexts},
$\VC$, with the following grammars:
  \[ C := x \mid [\cdot] \mid C_1 C_2 \mid \lambda x.C
\qquad\qquad \VC = [\cdot] \mid \lambda x.C\enspace. \]
}
\end{definition}

A context can contain multiple holes, which we label $[\cdot]_1,
[\cdot]_2, \dotsc, [\cdot]_n$; by convention, each hole is assigned a
unique $i$, and these are assigned in sequential order from left to
right.
Then if $C$ is a context with $n$ holes and $\wt{M}\in \Lambda^n$,
$C[\wt{M}] \in \Lambda$ is obtained by replacing $[\cdot]_i$ in $C$
with $M_i$, the $i$th projection of $\wt{M}$.
We further \arxiv{adopt the notation}{write} $C\fl{M}$ to denote the context $C$
whose every hole is filled with the term $M$.

\begin{definition}
  If ${\R} \subseteq \Lambda\times\Lambda$, then its \textbf{open
    contextual closure}, $\R^\oclos$, is given by \[{\R^\oclos} =
  \{\,(C[\wt{M}], C[\wt{N}]) \;|\; C\text{ is a context and }
  \wt{M}\R\wt{N}\,\},\] and its \textbf{closed contextual closure} is
  ${\R^\star} = {\R^\oclos} \cap {\Ld\times\Ld}$.
\end{definition}

Since we deal mostly with relations on closed terms, unless otherwise
specified, we take contextual closure to be closed contextual closure.

\arxiv{  
The following {two facts} {are immediate} from the definition of
contextual closure of binary relations on $\lambda$-terms.

\begin{lemma}
  The contextual closures of the empty set are the respective identity
  relations, i.e., $\emptyset^\oclos = \Id_{\Lambda\times\Lambda}$ and
  $\emptyset^\star = \Id_{\Ld\times\Ld}$.
\end{lemma}
}{
We observe that $\emptyset^\oclos = \Id_{\Lambda\times\Lambda}$ and
  $\emptyset^\star = \Id_{\Ld\times\Ld}$, where $\Id$ is the identity relation.
  We also have:
}

\begin{lemma}
  \label{lemma:rccomb}
  For ${\vartriangle} \in \{\oclos,\star\}$, if $A\R^\vartriangle C$ and $B \R^\vartriangle
  D$, then $AB \R^\vartriangle CD$.
\end{lemma}



The following technical lemma will frequently be used in the soundness
proofs for the \textit{up-to context} technique and when proving that
coupled logical relations are congruences.

\begin{lemma}
  \label{lemma:starsubst}
  For all ${\R} \subseteq \Ld\times\Ld$,
  $\wt{M} \R \wt{N}$, $M, N \in \Lambda$, and contexts $C_1$:
  \begin{enumerate}
  \item if $C_1[\wt{M}] = M \R^\star N = C_1[\wt{N}]$, then for all
    ${\vartriangle} \in \{\oclos,\star\}$, $E$ and $F$ such that
    $E\R^\vartriangle F$, and variables $x$, we have $M[E/x]\R^\star
    N[F/x]$;
  \item if $C_1[\wt{M}] = M \R^\oclos N = C_1[\wt{N}]$ with $\fv(M) =
    \fv(N) = \{x\}$ for some $x$, then for all ${\vartriangle} \in
    \{\oclos,\star\}$ and $E$ and $F$ such that $E\R^\vartriangle F$,
    we have $M[E/x]\R^\vartriangle N[F/x]$;
  \item if $C_1[\wt{M}] = M \R^\oclos N = C_1[\wt{N}]$, then for all
    ${\vartriangle} \in \{\oclos,\star\}$, $E$ and $F$ such that
    $E\R^\vartriangle F$, and variables $x$, we have $M[E/x]\R^\oclos
    N[F/x]$.
  \end{enumerate}

\end{lemma}

\arxiv{
  \begin{proof}
    Let $\wt{E}$, $\wt{F}$, and $C_2$ be such that $F = C_2[\wt{F}]$,
    and $\wt{E} \R \wt{F}$.
    Since each entry of $\wt{M}$ and $\wt{N}$ is in $\Ld$, $x \notin
    \fv(\wt{M})\cup\fv(\wt{N})$.
    Assume first that $M \R^\star N$, then $\fv(M) \cup \fv(N) =
    \emptyset$, so $M[E/x] = M$ and $N[F/x] = N$ for all $E, F$, and
    $x$.
    Thus, $M[E/x] = M \R^\star N = N[F/x]$ as desired.
    Now assume that $M \R^\circ N$.
    If $C_1 = [\cdot]$, then $M = C_1[\wt{M}] \in \Ld$ and $N =
    C_1[\wt{N}] \in \Ld$, so $M[E/x] \R N[E/x]$.
    Since ${\R} \subseteq {\R^\star} \subseteq {\R^\oclos}$, we
    conclude that $M[E/x] \R^\bull N[F/x]$.
    Now assume that $C_1 \neq [\cdot]$, then $(C_1[\wt{M}])[E/x] =
    (C_1[E/x])[\wt{M}]$ and similarly for $\wt{N}$ and $F$.
    Now consider the context $C_3 = C_1[C_2/x]$, or more formally, the
    context $C_3 = \gamma(C_1, x, C_2)$ where $\gamma(C_1, x, C_2)$ is
    recursively given by \[ \gamma(C_1,x,C_2) = \begin{cases}
      C_2 & \text{if } C_1 = x\\
      y & \text{if } C_1 = y\\
      [\cdot] & \text{if } C_1 = [\cdot]\\
      \lambda z.\gamma(C_1', x, C_2) & \text{if } C_1 = \lambda
      z.C_1'\text{ and } z\neq x\\
      \gamma(C_1', x, C_2)\gamma(C_1'', x, C_2) & \text{if } C_1 =
      C_1'C_1'',
    \end{cases}\] and let $\wt{M'}$ be the vector obtained by
    interweaving $\wt{M}$ and $\wt{E}$ such that $C_3[\wt{M'}] =
    C_1[C_2[\wt{E}]/x][\wt{M}]$.
    More formally, we can let $\wt{M'} = \phi(C_1,x,\wt{M},\wt{E})$
    where \[ \phi(C,x,\wt{M},\wt{E}) =
    \begin{cases}
      \wt{E} & \text{if } C = x\\
      y & \text{if } C = y\\
      M_i & \text{if } C = [\cdot]_i\\
      \phi(C',x,\wt{M},\wt{E}) & \text{if } C = \lambda z.C'\\
      \phi(C',x,\wt{M},\wt{E}) \mathrel{+\!+} \phi(C'',x,\wt{M},\wt{E})
      & \text{if } C = C'C''.
    \end{cases}\]

    Similarly, let $\wt{N'} = \phi(C,x,\wt{N},\wt{F})$, then $N[F/x] =
    C_3[\wt{N'}]$.
    If $\fv(M) = \fv(N) = \{x\}$, then $C_3[\wt{M'}]$ and
    $C_3[\wt{N'}]$ are closed, so $C_3[\wt{M'}] \R^\star
    C_3[\wt{N'}]$.
    In all other cases, since $C_3[\wt{M'}] \R^\oclos C_3[\wt{N'}]$,
    this completes our case analysis and we conclude the lemma.
  \end{proof}
}{}

\begin{corollary}
  \label{cor:rcredux}
  Suppose ${\R} \subseteq \Ld\times\Ld$, $\lambda x.C[\wt{P}] \R^\star
  \lambda x.C[\wt{Q}]$, and $M \R^\star N$.
  Then, we have $(C[\wt{P}])[M/x]
  \R^\star (C[\wt{Q}])[N/x]$.
\end{corollary}


We use the definition of congruence given by Selinger
\textcite{Selinger2008}:

\begin{definition}
  A relation $\R\ \subseteq \Lambda\times\Lambda$ is said to be a
  \textbf{congruence} if it is an equivalence relation and
  additionally respects the rules for constructing $\lambda$-terms,
  i.e., if it satisfies:
  \[ \frac{A\R C\quad B\R D}{AB \R CD}\quad\quad \frac{M \R N}{\lambda
    x.M \R \lambda x.N}.\]
\end{definition}








\arxiv{Finally, the}{The} following definitions will serve to define coupled
logical bisimulations in the next section:

\begin{definition}
  A \textbf{paired relation} $\R$ is a pair of relations $(\R_1,
  \R_2)$ with ${\R_1}, {\R_2} \subseteq \Ld\times\Ld$.
  A \textbf{coupled relation} is a paired relation $\R$ such that
  ${\R_1} \subseteq {\R_2}$.
\end{definition}

We define the usual set theoretic operations on coupled relations in a
pointwise manner\arxiv{, e.g.,}{:} for two coupled relations $\R, \R'$, ${\R'}
\subseteq {\R}$ if ${\R'_1} \subseteq {\R_1}$ and ${\R'_2} \subseteq
{\R_2}$, we let ${\R} \cup {\R'} = ({\R_1} \cup {\R'_1}, {\R_2} \cup
{\R'_2})$, etc.

\section{CLB in the Call-by-name $\lambda$-calculus}
\label{s:cbn}

We begin by considering the theory of coupled logical bisimulations
for the call-by-name (cbn) $\lambda$-calculus.
Many proofs are ommited or only sketched since they can be seen as
simplifications of the proofs given in the call-by-value case
(Section~\ref{s:cbv}).

\begin{definition}
  The \textbf{call-by-name} $\lambda$-calculus is defined by the
  following reduction rules:
  \begin{equation*}
    \frac{M\ra M'}{MN\ra M'N}\hspace{3em}
    \frac{}{(\lambda x.P)N \ra P[N/x]}.
  \end{equation*}
  We take the set $\V$ of \textbf{values} to be the set of all
  abstractions $\lambda x.P \in \Ld$.

  We write $\Ra$ for the reflexive and transitive closure of $\ra$,
  and say that a term $M$ \textbf{converges}, written $M\conv$, if
  there exists a value $\lambda x.P$ such that $M \Ra \lambda x.P$; we
  may also write $M \Downarrow \lambda x.P$ in this case.
  Similarly, we say $M$ \textbf{diverges}, written $M\diver$, if it
  does not converge; in this case, it will sometimes be useful to
  write $M \Uparrow M'$ if $M \Ra M'$ and $M'\diver$.
\end{definition}

\begin{definition}
  A
  \textbf{cbn evaluation context}
  $\EC$ is given by
  \arxiv{}{$ \EC := [\cdot] \mid \EC M $}
  \arxiv{the
    following grammar}{}
  ($M$ ranges over $\Ld$)\arxiv{:}{.}
  \arxiv{\[ \EC := [\cdot] \mid \EC M\enspace. \]}{}
\end{definition}

\begin{definition}
  Two terms $M, N \in \Ld$ are \textbf{contextually equivalent},
  written $M \cen N$, if for all contexts $C$, $C[M]\conv$ if and only
  if $C[N]\conv$. Similarly, we say that two terms $M$ and $N$ are
  \textbf{evaluation-contextually equivalent}, written $M \ecen N$, if
  for all \arxiv{evaluation contexts}{} $\EC$, $\EC[M]\conv$ if and only if
  $\EC[N]\conv$.
\end{definition}

Although contextual equivalence appears to be \arxiv{a stronger notion
of equivalence}{stronger} than evaluation-contextual equivalence, Milner's
\cite{Milner1977} ``context lemma'' (see also Gordon \cite[Proposition
4.18]{Gordon1994}) tells us otherwise:

\begin{theorem}[\cite{Milner1977}]\label{t:contextlemma:cbn}
  \label{thm:eceisce}
  Evaluation-contextual equivalence and contextual equivalence
  coincide, i.e., $M \ecen N$ if and only if $M \cen N$.
\end{theorem}

The following definition is based on \textcite[Lemma
5.8]{DalLago2014}:

\begin{definition}
  If $\R, \R'\ \subseteq \Ld\times\Ld $ are relations, then the
  \textbf{evaluation-contextual closure of $\R$ under $\R'$}, written $\eccn{\R}{\R'}$,
  is given by
  \[{\eccnnp{\R}{\R'}} = \{\, (E\wt{M}, F\wt{N}) \;|\; E\R F \text{
    and } \wt{M}\R'\wt{N} \,\}\] where $\wt{M}$ and $\wt{N}$ are
  potentially empty.
\end{definition}


We remark that for all $\R$ and $\R'$, ${\R} \subseteq
{\eccn{\R}{\R'}}$.

\begin{definition}
  If R is a coupled relation, then let its \textbf{contextual closure}
  $\R^C$ be given by
  \[\R^C = \left(\R_1^\star, {\eccn{\R_2}{\R_1^\star}} \cup
    {\R_1^\star} \right).\]
\end{definition}

\subsection{Coupled Logical Bisimulation}

The following definition is extracted from the corresponding notion in
\cite{DalLago2014}:
\begin{definition}\label{def:clb}
  A coupled relation $\R$ is a \textbf{coupled logical bisimulation}
  (CLB) if \arxiv{ whenever $M \R_2 N$, we have:}{$M \R_2 N$ implies:}
  \begin{enumerate}
  \item if $M \ra M'$, then there exists an $N'$ such that $N \Ra N'$
    and $M' \R_2 N'$;
  \item if $M = \lambda x.M'$, then $N \Ra \lambda x.N'$ and for all
    $P, Q \in \Ld$ such that $P \R_1^\star Q$, we have $M'
    [P/x]\R_2N'[Q/x]$;
  \item   and  the converses of the two previous conditions for $N$.
  \end{enumerate}
  Coupled logical bisimilarity, written $\clbn\ = (\clbn_1, \clbn_2)$,
  is the pairwise union of all CLBs.
\end{definition}


As one would hope, CLBs have a continuous progression:

\begin{definition}[CLB progressions, call-by-name case]
  \label{def:progresses}
  Given coupled relations $\R$ and $\SR$, we say \textbf{$\R$
    progresses to $\SR$}, written ${R} \prog {\SR}$, if whenever $M
  \R_2 N$, then:
  \begin{enumerate}
  \item whenever $M \ra M'$ then $N \Ra N'$ and $M' \SR_2 N'$;
  \item whenever $M = \lambda x.P$ then $N \Ra \lambda x.Q$ such that
    for all $X \R_1^\star Y$, $P[X/x] \SR_2 Q[Y/x]$;
  \item the converses of the previous two conditions for $N$.
  \end{enumerate}
\end{definition}

\begin{prop}
  \label{prop:clbnprog}
  For all coupled relations $\R$, ${\R} \prog {\R}$ if and only if
  $\R$ is a CLB. 
  Thus, the relation $\prog$ is a progression for CLBs in the universe
  of coupled relations.
\end{prop}

\begin{prop}
  \label{prop:clbnsympa}
  The relation $\prog$ is continuous.
\end{prop}

\begin{proof}
  \newcommand{\nr}{\mathrel{(\nu{\R})}}
  \newcommand{\ns}{\mathrel{(\nu{\SR})}}
  \newcommand{\rn}{\mathrel{({\R}_n)}}
  \newcommand{\sn}{\mathrel{({\SR}_n)}}
  \arxiv{
    Assume ${\R_0} \subseteq {\R_1} \subseteq \cdots$ and ${\SR_0}
    \subseteq {\SR_1} \subseteq \cdots$ are two ascending chains of
    relations such that ${\R_i} \prog {\SR_i}$ for all $i$, and assume
    $M \nr_2 N$ and $X \nr_1^\star Y$ for arbitrary $M, N, X, Y$.
    Then $X = C[\wt{M}]$ and $Y = C[\wt{N}]$ for some $\wt{M} \nr_1
    \wt{N}$.
    Moreover, there exists a $K$ such that for all $n > K$, $M \rn_2 N$
    and $\wt{M} \rn_1 \wt{N}$; fix any such $n > K$.

    If $M \ra M'$, then, since ${\R_n} \prog {\SR_n}$, $N \Ra N'$ such
    that $M' \sn_2 N'$.
    Since ${\SR_n} \subseteq \nu{\SR}$, we then get that $M' \ns_2 N'$
    and we are done.

    If $M = \lambda x.M'$, then, since ${\R_n} \prog {\SR_n}$, $N \Ra
    \lambda x.N'$ such that $M'[X/x] \sn_2 N'[Y/x]$ for all $X
    \rn_1^\star Y$.
    Since ${\SR_n} \subseteq \nu{\SR}$, we then get that $M'[X/x]
    \ns_2 N'[Y/x]$ and we are done.

    The symmetric cases for $N$ follow symmetrically.
  }{
    Assume ${\R_0} \subseteq {\R_1} \subseteq \cdots$ and ${\SR_0}
    \subseteq {\SR_1} \subseteq \cdots$ are two ascending chains of
    relations such that ${\R_i} \prog {\SR_i}$ for all $i$, and assume
    $M \nr_2 N$ for arbitrary $M, N$.
    We establish the result by case analysis on whether $M$ and $N$
    are values or can reduce further.
  }
\end{proof}

In the notation of Section \ref{sec:genupto}, the family of desired
relations $\F$ is the set of all CLBs, and the universe $\U$ of
relations we are working in is the set of all coupled relations.
Unfortunately, we cannot extend our notion of progression to all
paired relations since it is not the case that for all paired relations $\R$,
${\R} \prog {\R}$ only if there exists a CLB $\R'$ such that ${\R}
\subseteq {\R'}$.
To see why, consider the paired relation ${\R} = \left( \left\{ \left(
      I, \Omega \right) \right\}, \left\{ \left( \Omega, \Omega
    \right) \right\} \right)$, where $\Omega = (\lambda x.xx)(\lambda
x.xx)$ \arxiv{is the standard divergent term}{} and $I = \lambda x.x$\arxiv{ the
identity.}{.}
It is easy to see that ${\R} \prog {\R}$, so assume there existed a
CLB $\R'$ such that ${\R} \subseteq {\R'}$.
Since every CLB is a coupled relation, 
we would have $(I,
\Omega) \in {\R_1} \subseteq {\R'_1} \subseteq {\R'_2}$.
\arxiv{Since $I = \lambda x.x$, by}{By} the second clause of the definition of CLB, $\Omega$ would converge.

We further develop properties of progression that will be useful in
showing the soundness of various up-to techniques.

\begin{prop}
  \label{prop:clbnprogcuiuu}
  \newcommand{\Ri}{\mathrel{({\R_i})}}
  \newcommand{\SRi}{\mathrel{({\SR_i})}}

  If $\left\{{\R_i}\right\}_{i\in I}$ and
  $\left\{{\SR_i}\right\}_{i\in I}$ are two families of paired
  relations such that ${\R_i} \prog {\SR_i}$ for all $i \in I$, then
  $\left(\bigcap_{i\in I} {\Ri_1}, \bigcup_{i\in I} {\Ri_2}\right)
  \prog \bigcup_{i\in I} {\SR_i}$.
\end{prop}

\begin{proof}
  \renewcommand{\SS}{\mathrel{\left(\bigcup_{i\in I} {\SR_i}\right)}}
  \newcommand{\UU}{\mathrel{\U}}
  \newcommand{\Ri}{\mathrel{({\R_i})}}
  \newcommand{\SRi}{\mathrel{({\SR_i})}}

  Let ${\UU} = (\bigcap_{i\in I} \Ri_1, \bigcup_{i\in I} \Ri_2)$, and
  assume $M \UU_2 N$.
  Then there exists some $\R_i$ such that $M \Ri_2 N$.

  If $M \ra M'$, then $N \Ra N'$ such that $M' \SRi_2 N'$, and by
  inclusion, we get that $M' \SS_2 N'$.

  If $M = \lambda x.M'$ and $X \UU_1^\star Y$, then since ${\UU_1}
  \subseteq {\Ri_1}$, we have $X \Ri_1^\star Y$ by the monotonicity of
  contextual closure.
  Then since ${\R_i} \prog {\SR_i}$, $N \Ra \lambda x.N'$ and $M'[X/x]
  \SRi_2 N'[Y/x]$, and by inclusion, we get that $M'[X/x] \SS_2
  N'[Y/x]$.

  The symmetric cases for $N$ follow symmetrically, and we're done.
\end{proof}

\begin{prop}
  Progression $\prog$ is closed under left intersection and right
  union.
\end{prop}

\begin{prop}
  \label{prop:r1starcompat}
  For all $\R,\SR$ 
  such that ${\R_1^\star} \subseteq {\R_2}$ and
  ${\SR_1^\star} \subseteq {\SR_2}$, if ${\R} \prog {\SR}$ then
  $({\R_1^\star}, {\R_2}) \prog ({\SR_1^\star}, {\SR_2})$.
\end{prop}

\arxiv{
\begin{proof}
  Immediate by the idempotence of contextual closure.
\end{proof}
}{}

\begin{corollary}
  \label{cor:clbstar}
  If $\R$ is a CLB such that ${\R_1^\star} \subseteq {\R_2}$, then
  $(\R_1^\star, \R_2)$ is also a CLB.
\end{corollary}

\begin{proof}
  By Proposition \ref{prop:clbnprog}, ${\R} \prog {\R}$, so by
  Proposition \ref{prop:r1starcompat}, $({\R_1^\star}, {\R_2}) \prog
  ({\R_1^\star}, {\R_2})$ and $({\R_1^\star}, {\R_2})$ is a coupled
  relation.
  Then again, by Proposition \ref{prop:clbnprog}, $({\R_1^\star},
  {\R_2})$ is a CLB.
\end{proof}

\begin{lemma}
  \label{lemma:clbsubs}
  If $(\R_1, \R_2)$ is a CLB, then so is $(\R_1', \R_2)$ for all
  $\R_1'\ \subseteq\ \R_1$.
\end{lemma}

\arxiv{
\begin{proof}
  Check the definition and use the monotonicity of contextual
  closure.
\end{proof}
}{}

\begin{corollary}
  If $(\R_1, \R_2)$ is a CLB, then so is $(\emptyset, \R_2)$.
\end{corollary}

\begin{lemma}
  Let $\mathcal{A}, \mathcal{B} \subseteq\ \R_2$ such that
  $(\mathcal{A}, \R_2)$ and $(\mathcal{B}, \R_2)$ are \arxiv{both}{} CLBs.
  Then $(\mathcal{A}\cup\mathcal{B},\R_2)$ is \arxiv{also}{}
  a CLB.
\end{lemma}

\arxiv{
\begin{proof}
  Follows from a straightforward check of the definition.
\end{proof}
}{}

\begin{lemma}
  \label{lemma:bigsteps}
  If $\R$ is a CLB, $M\R_2N$ and $M \Ra M'$, then there exists $N'
  \in \Ld$ such that $N \Ra N'$ and $M' \R_2 N'$.
\end{lemma}

\begin{proof}
  Follows by a straightforward induction on the length of the
  reduction $M \Ra M'$.
\end{proof}

{Having introduced the necessary material for CLB, we can work
  with the theory of Section~\ref{sec:genupto} to analyse up-to
  techniques for CLB in the cbn case. We study two versions of the
  up-to contexts technique.
This will allow us to deduce congruence properties for CLB.  
}

\begin{definition}
  We call \textbf{up-to evaluation context} the up-to technique $\Pev$
  given by:
  \[ \Pev({\R}) = \left(\R_1, {\R_2} \cup \left\{(EM, FN) \mid E \R_2
      F, M \R_1^\star N \right\}\right). \]
\end{definition}

\begin{prop}
  \label{prop:nupeveccn}
  We have that $\nu\Pev({\R}) = (\R_1, \eccn{\R_2}{\R_1^\star})$.
\end{prop}

\arxiv{
\begin{proof}
  Simple double inclusion.
\end{proof}
}{}

\begin{prop}
  \label{prop:UTEextrespc}
  Up-to evaluation context is extensive and respectfully compatible.
\end{prop}

\begin{proof}
  \newcommand{\pr}{\mathrel{\Pev({\R})}}
  \newcommand{\ps}{\mathrel{\Pev({\SR})}}

  \arxiv{
    Extensiveness is immediate by definition, so assume ${\R} \prog
    {\SR}$ and ${\R} \subseteq {\SR}$.
    Clearly $\Pev({\R}) \subseteq \Pev({\SR})$.
    We want to show that $\Pev({\R}) \prog \Pev({\SR})$, so assume $M
    \pr_2 N$.
    If $M \R_2 N$, then we're done since ${\R} \prog {\SR}$, so assume
    instead that $M = EM'$ and $N = FN'$ for $E \R_2 F$ and $M'
    \R_1^\star N'$.
    Since ${\R} \subseteq {\SR}$ and contextual closure is monotone, we
    have $M' \SR_1^\star N'$.

    Assume first that $M \ra M''$, then we fall into two cases.
    The first is that $EM' \ra E'M'$.
    Since $E \R_2 F$, then $F \Ra F'$ such that $E' \R_2 F'$.
    Then $E'M' \ps_2 F'N'$ and we're done.
    The second is that $E = \lambda x.E'$ and that $EM' \ra E'[M'/x]$.
    Then since $E \R_2 F$, then $F \Ra \lambda x.F'$ such that $\lambda
    x.E' \SR_2 \lambda x.F'$ and for all $X \R_1^\star Y$, $E'[X/x]
    \SR_2 F'[Y/x]$.
    Thus, $FN' \Ra F'[N'/x]$ and $E'[M'/x] \ps_2 F'[N'/x]$ and we're
    done.

    The case that $EM' = \lambda x.M''$ is impossible, so we're done
    and conclude respectful compatibility.
  }{%
    Extensiveness is immediate by the definition of up-to evaluation
    context, and respectful compatibility follows from a relatively
    straightforward check of the definition.
  }
\end{proof}

\begin{corollary}
  Up-to evaluation context is sound.
\arxiv{\end{corollary}

\begin{corollary}
  \label{cor:clbevclos}
  If}{

Moreover, if} $\R$ is a CLB, then so is $(\R_1, \eccn{\R_2}{\R_1^\star})$.
\end{corollary}

\arxiv{
  \begin{proof}
    By Proposition \ref{prop:clbnprog}, ${\R} \prog {\R}$.
    Then by compatibility, $\Pev^i({\R}) \prog \Pev^i({\R})$ for all $i
    \in \N$, and by extensiveness, we have an ascending chain of relations
    ${\R} \subseteq \Pev({\R}) \subseteq \Pev^2({\R}) \subseteq \cdots$.
    Then by continuity (Proposition \ref{prop:clbnsympa}),
    $\nu\Pev({\R}) \prog \nu\Pev({\R})$, and so by Proposition
    \ref{prop:clbnprog}, $\nu\Pev({\R})$ is a CLB.
    Since $\nu\Pev({\R}) = (\R_1, \eccn{\R_2}{\R_1^\star})$, we conclude
    that $(\R_1, \eccn{\R_2}{\R_1^\star})$ is a CLB.
  \end{proof}
}{}

\begin{definition}
  We call \textbf{up-to context} the up-to technique given by ${\R}
  \mapsto {\R}^C$.
\end{definition}

\arxiv{
  \begin{lemma}
    \label{lemma:rc2comb}
    If $E \R^C_2 F$ and $\wt{M} \R_1^\star \wt{N}$, then $E\wt{M} \R^C_2
    F\wt{N}$.
  \end{lemma}

  \begin{proof}
    By case analysis on $E \R^C_2 F$.
    If $E \R^C_2 F$ because $E \R_1^\star F$, then the conclusion
    follows by Lemma \ref{lemma:rccomb}.
    Finally, if $E \R^C_2 F$ because there exist $E', F', \wt{K},
    \wt{L}$ such that $E' \R_2 F'$, $\wt{K} R_1^\star \wt{L}$, $E =
    E'\wt{K}$, and $F = F'\wt{L}$, then by Lemma \ref{lemma:rccomb},
    $\wt{K}\wt{M}\R_1^\star \wt{L}\wt{N}$, and so by construction,
    $E\wt{M} = E' \wt{K}\wt{M} R^C_2 F'\wt{L}\wt{N} = F\wt{N}$.
  \end{proof}
}{}

\begin{prop}
  Up-to context is extensive and respectfully compatible.
\end{prop}

\begin{proof}
  Extensiveness is immediate by definition of up-to context.

  Assume ${\R} \prog {\SR}$ and ${\R} \subseteq {\SR}$; we want to
  show that ${\R^C} \subseteq {\SR^C}$ and ${\R^C} \prog {\SR^C}$.
  The inclusion ${\R^C} \subseteq {\SR^C}$ is obvious.
  Since $(\R_1, \eccn{\R_2}{\R_1^\star}) \prog (\SR_1,
  \eccn{\SR_2}{\SR_1^\star})$ by Propositions \ref{prop:UTEextrespc}
  and \ref{prop:nupeveccn} and Corollary \ref{cor:nupsound}, it is
  sufficient to show that $(\R_1, \R_1^\star \setminus
  \eccn{\R_2}{\R_1^\star}) \prog {\SR^C}$, since \arxiv{%
    then, by Proposition \ref{prop:clbnprogcuiuu},
    \[ \left( {\R_1} \cap {\R_1}, \eccn{\R_2}{\R_1^\star} \cup
      \left(\R_1^\star \setminus \eccn{\R_2}{\R_1^\star} \right)
    \right ) \prog \left(\SR \cup \SR^C\right), \] i.e., $\left( \R_1,
      \R^C_2 \right) \prog {\SR^C}$ by extensiveness and Proposition
    \ref{prop:nupeveccn}.
  }{%
    then $\left( \R_1, \R^C_2 \right) \prog {\SR^C}$ by Proposition
    \ref{prop:clbnprogcuiuu}, extensiveness, and Proposition
    \ref{prop:nupeveccn}.
  } Then, since contextual closure is idempotent and ${\R_1^\star}
  \subseteq {\R_2^C}$ and $({\SR_1^C})^\star \subseteq {\SR^C_2}$, we
  get $\left(\R_1^\star, \R^C_2 \right) \prog {\SR^C}$ by Proposition
  \ref{prop:r1starcompat}, i.e., ${\R^C} \prog {\SR^C}$.

  \arxiv{
    We show that $(\R_1, \R_1^\star \setminus \eccn{\R_2}{\R_1^\star})
    \prog {\SR^C}$.
    Assume $M\R_1^\star N$.
    Then $M=C[\wt{M}]$ and $N = C[\wt{N}]$ for $\wt{M}\R_1\wt{N}$.
    We proceed by induction on $C$.

    If $C = [\cdot]$, then $M = M_1 \R_1 N_1 = N$, and since $R_1
    \subseteq R_2$, and ${\R} \prog {\SR} \subseteq {\SR^C}$ and $\prog$
    is closed under right union, we're done.

    The case $C=x$ is impossible: it is never the case that $x
    \R_1^\star x$ since $x \notin \Ld$.

    Consider the case where $C = \lambda x.C'$, i.e., $\lambda
    x.C'[\wt{M}] \R_1^\star \lambda x.C'[\wt{N}]$.
    Since there exists no $M'$ such that $\lambda x.C'[\wt{M}] \ra M'$,
    we need only check the second clause of the definition of CLB,
    namely that for all $X \R_1^\star Y$, $(C'[\wt{M}])[X/x] \SR^C_2
    (C'[\wt{N}])[Y/x]$.
    By Corollary \ref{cor:rcredux}, $(C'[\wt{M}])[X/x] \R_1^\star
    (C'[\wt{N}])[Y/x]$, so ${\R_1^\star} \subseteq {\SR_1^\star}
    \subseteq {\SR^C_2}$, $(C'[\wt{M}])[X/x] \SR^C_2 (C'[\wt{N}])[Y/x]$
    and we're done.

    Now assume that $C = C_1 C_2$, and let $\wt{M} = \wt{M_1}\wt{M_2}$
    be such that $C[\wt{M}] = C_1[\wt{M_1}]C_2[\wt{M_2}]$ and similarly
    for $\wt{N}$.
    We show the first part of the definition, namely, that if $M\ra M'$
    then $N \Ra N'$ such that $M'\SR^C_2N'$.
    We fall into one of four mutually exclusive cases:
    \begin{enumerate}[(i)]
    \item $C_1[\wt{M_1}] \ra C_1'[\wt{M_1}]$
    \item $C_1[\wt{M_1}] \ra C_1[\wt{M_1'}]$
    \item $C_1[\wt{M_1}] \ra_\beta M_1'$
    \item $C_1[\wt{M_1}] = \lambda x.P$ and $M \ra_\beta P
      [C_2[\wt{M_2}]/x]$.
    \end{enumerate}

    If we fall into case (i), then there is a transition internal to
    the context $C_1$ not involving any of the $M_i$, that is to say,
    ``$C_1[\cdot] \ra C_1'[\cdot]$''.
    Then $M \ra C_1'[\wt{M_1}]C_2[\wt{M_2}]$ and $N \Ra
    C_1'[\wt{N_1}]C_2[\wt{N_2}]$.
    Since $C_1'C_2$ is again a context, we have
    \[C_1'[\wt{M_1}]C_2[\wt{M_2}] \R_1^\star
    C_1'[\wt{N_1}]C_2[\wt{N_2}],\] and since ${\R_1^\star} \subseteq
    {\SR_1^\star} \subseteq {\SR^C_2}$,
    \[C_1'[\wt{M_1}]C_2[\wt{M_2}] \SR^C_2 C_1'[\wt{N_1}]C_2[\wt{N_2}]\]
    as desired.

    Case (ii) is impossible: take $\wt{M_1} = (M_1,\dotsc,M_n)$ and
    $\wt{N_1} = (N_1,\dotsc,N_n)$. If $C_1[\wt{M_1}]\ra C_1[\wt{M_1'}]$,
    then by call-by-name reduction implies that $C_1$ is of the form
    $C_1 = [\cdot]C_1'$.
    Thus, $M = [M_1]C'_1[M_2,\dotsc,M_n]C_2[\wt{M_2}]
    \eccn{\R_2}{\R_1^\star} [N_1]C'_1[N_2,\dotsc,N_n]C_2[\wt{N_2}] = N$,
    a contradiction.

    If we fall into case (iii), that's to say, if there's a
    $\beta$-reduction involving the context $C_1$ and $M_1$, then
    $C_1[\wt{M}] \ra_\beta M_1'$, so $M \ra M_1'C_2[\wt{M_2}]$.
    By the induction hypothesis, $C_1[\wt{N}] \Ra N_1'$ such that $M_1'
    \SR^C_2 N_1'$. Thus, $N \Ra N_1'C_2[\wt{N_2}]$.
    Since $C_2[\wt{M_2}] \R_1^\star C_2[\wt{N_2}]$ and ${\R_1^\star}
    \subseteq {\SR_1^\star}$, by Lemma \ref{lemma:rc2comb}, we conclude
    that $M_1'C_2[\wt{M_2}] \SR^C_2 N_1'C_2[\wt{N_2}]$ as desired.

    Finally, we consider case (iv).
    By the induction hypothesis, $C_1[\wt{N_1}] \Ra \lambda x.Q$ such
    that $\lambda x.P \SR_2^C \lambda x.Q$ and for all $X \R_1^\star Y$,
    $P[X/x] \SR^C_2 Q[Y/x]$.
    Thus, since $C_2[\wt{M_2}] \R_1^\star C_1[\wt{N_2}]$, we get
    $P[C_2[\wt{M_2}]/x] \SR^C_2 Q[C_2[\wt{N_2}]/x]$ as desired.

    This completes the induction on $C$.
    Thus, we conclude that $(\R_1, \R_1^\star \setminus
    \eccn{\R_2}{\R_1^\star}) \prog {\SR^C}$, and thus that ${\R^C}
    \prog {\SR^C}$.
  }{ To show that $(\R_1, \R_1^\star \setminus
    \eccn{\R_2}{\R_1^\star}) \prog {\SR^C}$, assume $M \R_1^\star N$,
    where $M=C[\wt{M}]$ and $N = C[\wt{N}]$ for some $C$ and
    $\wt{M}\R_1\wt{N}$, and proceed by structural induction on $C$.
  }
\end{proof}



\begin{corollary}
  Up-to context is sound.
\arxiv{
\end{corollary}

\begin{corollary}
  If}{

Moreover, if} $\R$ is a CLB, then so is $\R^C$.
\end{corollary}

\arxiv{
  \begin{proof}
    Immediate by Proposition \ref{prop:clbnprog} and respectful
    compatibility.
  \end{proof}
}{}

\begin{corollary}
  \leavevmode
  \label{corollary:Clift}
  \begin{enumerate}
  \item If $M \clbn_1 N$, then $C[M] \clbn_1 C[N]$ for all contexts
    $C$.
  \item If $E \clbn_2 F$, then $\EC[E] \clbn_2 \EC[F]$ for all
    evaluation contexts $\EC$.
  \end{enumerate}
\end{corollary}

\arxiv{
\begin{proof}
  We show the first statement.
  If $M \clbn_1 M$, then there exists a CLB ${\R} \subseteq {\clbn}$
  such that $M\R_1 N$.
  Then ${\R^C} \subseteq {\clbn}$ also and ${\R_1^\star} = {\R^C_1}
  \subseteq {\clbn_1}$.

  We proceed in a similar fashion to show the second statement: if $E
  \clbn_2 F$ then there exists a CLB ${\R} \subseteq {\clbn}$ such
  that $E \R_2 F$. Then again, ${\R^C} \subseteq\ {\clbn}$.
  Since ${\eccn{\R_2}{\R_1^\star}} \subseteq {\R^C_2}$ and ${\Id}
  \subseteq {\R_1^\star}$, we have ${\eccn{R_2}{\Id}} \subseteq
  {\R^C_2}$.
  Thus, $E \R_2 F$ implies $\EC[E] \R^C_2 \EC[F]$ for all evaluation
  contexts $\EC$, and since ${R^C_2} \subseteq {\clbn_2}$, we conclude
  the second statement.
\end{proof}
}{}

\begin{corollary}
  \label{cor:inclusion1}
  We have the inclusion $\clbn\ \subseteq (\cen, \ecen)$, i.e., if $M
  \clbn_1 N$, then $M \cen N$, and if $E \clbn_2 F$, then $E \ece F$.
\end{corollary}

\arxiv{
\begin{proof}
  Assume first that $M \clbn_1 N$, then $C[M] \clbn_1 C[N]$, and so
  for some CLB $\R$, $C[M] \R C[N]$.
  By Lemma \ref{lemma:bigsteps}, this implies that $C[M] \Ra \lambda
  x.P$ for some $P$ if and only if $C[N] \Ra \lambda x.Q$ for some
  $Q$, and conversely.
  But this is exactly the definition of $M \cen N$.
  The case of $E \clbn_2 F$ follows in an identical manner.
\end{proof}
}{}

\arxiv{
\begin{lemma}
  \label{lemma:Raece}
  If $E \Ra E'$, then $E \ecen E'$.
\end{lemma}

\begin{proof}
  It is easy to verify that $(\emptyset, \Ra)$ is a CLB.
  By Corollary \ref{cor:inclusion1}, we get $\Ra\ \subseteq\ \clbn_2\
  \subseteq\ \ecen$, so if $E \Ra E'$, $E \ecen E'$.
\end{proof}

\begin{lemma}
  \label{lemma:eceabssubst}
  If $\lambda x.P \ecen \lambda x.Q$, then $P[M/x] \ecen Q[M/x]$ for
  all $M$.
\end{lemma}

\begin{proof}
  If $\lambda x.P \ecen \lambda x.Q$, then by Corollary
  \ref{corollary:Clift}, $(\lambda x.P)M \ecen (\lambda x.Q)M$ for all
  $M$.
  Moreover, $(\lambda x.P)M \Ra P[M/x]$ and $(\lambda x.Q)M \Ra
  Q[M/x]$.
  Thus, by Lemma \ref{lemma:Raece}, \[ P[M/x] \ecen (\lambda x.P)M
  \ece (\lambda x.Q)M \ecen Q[M/x] \] and the conclusion follows by
  transitivity of $\ecen$.
\end{proof}

\begin{lemma}
  \label{lemma:ecesubst}
  If $M \cen N$, then for all $P \in \Lambda$, $P[M/x] \ecen P[N/x]$
  for all $x \notin \bv(P)$.
\end{lemma}

\begin{proof}
  Let $C = P[[\cdot]/x]$ be the context obtained by replacing every
  occurence of $x$ with a hole, and so $C[M] = P[M/x]$ and similarly
  for $N$.
  Then, by Corollary \ref{corollary:Clift}, $C[M] \cen C[N]$, so
  $P[M/x] \cen P[N/x]$.
  Since $\cen\ \subseteq\ \ecen$, we conclude $P[M/x] \ecen P[N/x]$ as
  desired.
\end{proof}
}{}

\begin{prop}
  The coupled relation $(\cen, \ecen)$ is a CLB, that is to say,
  $(\cen, \ecen) \subseteq\ \clbn$.
\end{prop}

  \arxiv{
\begin{proof}

    We show the first clause of the definition of CLB.
    Assume $M \ecen N$ and that $M \ra M'$ (so $M \Ra M'$, and thus $M
    \ecen M'$).
    Then $N \Ra N$ and $M' \ecen N$ as desired follows by the
    transitivity and symmetry of $\ecen$.

    We now show that $\lambda x.P \ecen N$ satisfies the second clause
    of the definition of CLB.
    By definition of $\ecen$, $N \Ra \lambda x.Q$.
    Let $V \cen W$, then by Lemma \ref{lemma:eceabssubst}, we have
    $P[V/x] \ecen Q[V/x]$, and by Lemma \ref{lemma:ecesubst}, we have
    $Q[V/x] \ecen Q[W/x]$.
    By transitivity of $\ecen$, we conclude that $P[V/x] \ecen Q[W/x]$
    as desired.
\end{proof}
  }{
  }

\begin{corollary}
  \label{cor:censareclbns}
  The contextual equivalences and coupled logical bisimilarity
  coincide; that is, we have $(\cen, \ecen) = (\clbn_1, \clbn_2)$.
\end{corollary}

\begin{corollary}
  \label{cor:clbn1-eq-clbn2}
  The two components of coupled logical bisimilarity coincide with
  each other and with contextual equivalences, i.e., ${\clbn_1} =
  {\cen} = {\ecen} = {\clbn_2}$.
\end{corollary}

\begin{proof}
  Immediate by Corollary \ref{cor:censareclbns} and Theorem
  \ref{thm:eceisce}.
\end{proof}

\arxiv{
\subsubsection{Further up-to techniques}
\label{sss:clbnupto}

We present further up-to techniques, in addition to the up-to
evaluation and up-to context techniques presented above.

\begin{definition}
  We call \textbf{up-to reduction} the up-to technique ${\R} \mapsto
  {\utr{\R}}$, where $({\utr{\R}})_1 = {\R_1}$ and $({\utr{\R}})_2 =
  ({\Ra\R_2\Longleftarrow})$, i.e., $M \mathrel{({\utr{\R}})_2} N$ if
  there exist $M'$ and $N'$ such that $M \Ra M'$, $N \Ra N'$, and $M'
  \R_2 N'$.
\end{definition}

\begin{prop}
  Up-to reduction is extensive and respectfully compatible.
\end{prop}

\begin{proof} \newcommand{\UR}{\mathrel{({\utr{\R}})}}
  \newcommand{\US}{\mathrel{({\utr{\SR}})}}

  Extensiveness is immediate by the reflexivity of $\Ra$: by
  definition, ${\R_1} = {\UR_1}$.
  Moreover, $\R_2\ \subseteq\ \UR_2$ since whenever $M \R_2 N$, $M \Ra
  M \R_2 N \Longleftarrow N$, i.e., $M \UR_2 N$.
  Hence, ${\R} \subseteq {\UR}$.

  We now show compatibility.
  Assume ${\R} \prog {\SR}$ and ${\R} \subseteq {\S}$, then we want to
  show that ${\utr{\R}} \prog {\utr{\SR}}$.

  Assume $M \UR_2 N$, and let $M''$ and $N''$ be such that $M \Ra M''
  \R_2 N'' \Longleftarrow N$.
  First assume that $M \ra M'$.
  If $M \neq M''$, then we're done, since $N \Ra N$ and $M' \Ra M''
  \R_2 N'' \Longleftarrow N$, and by inclusion of ${\R}$ in ${\SR}$,
  $M' \Ra M'' \SR_2 N'' \Longleftarrow N$.
  Otherwise, if $M = M''$, then the fact that ${\R} \prog {\SR}$ and
  $M'' \R_2 N''$ implies that $N'' \Ra N'$ such that $M' \SR_2 N'$, so
  $N \Ra N'$ by transitivity of $\Ra$, and $M' \US_2 N'$ by
  reflexivity of $\Ra$, as desired.

  Now assume that $M = \lambda x.P$, then $M'' = \lambda x.P$ as well.
  Then since $\lambda x.P \R_2 N''$ and ${\R} \prog {\SR}$, we have
  that $N'' \Ra \lambda x.Q$ such that for all $X \R_1^\star Y$,
  $P[X/x] \SR_2 Q[Y/x]$.
  By transitivity of $\Ra$, we thus have that $N \Ra \lambda x.Q$ and
  for all $X \UR_1^\star Y$, that $P[X/x] \UR_2 Q[Y/x]$ as desired.
  \RAK{This last sentence needs to be proven}

  The symmetric cases follow symmetrically and we conclude respectful
  compatibility.
\end{proof}

\begin{corollary}
  Up-to reduction is sound.
\end{corollary}
}{}

\arxiv{
\subsection{Applicative Bisimulation}

We recall the big-step version of applicative bisimulation, as
originally presented by Abramsky \cite{Abramsky1990}:

\begin{definition}
  A relation $\R\ \subseteq \Ld\times\Ld$ is called an
  \textbf{applicative bisimulation} if $M\R N$ implies that whenever
  $M \Ra \lambda x.P$, $N \Ra \lambda x.Q$ for some $Q$, and $P[W/x]
  \R Q[W/x]$ for all $W \in \Ld$, and conversely for $N$.
  We call the union of all applicative bisimulations, written $\ABn$,
  \textbf{applicative bisimilarity}.
\end{definition}

It is not hard to show that applicative bisimilarity is itself an
applicative bisimulation.

\begin{prop}
  \label{prop:CLBAB}
  If $({\R'}, {\R})$ is a CLB for some ${\R'} \subseteq {\R} \cap
  {\Id}$, then $\R$ is an applicative bisimulation.
\end{prop}

\arxiv{
\begin{proof}
  Immediate by Lemma \ref{lemma:bigsteps} and the definitions of
  applicative bisimulation and coupled logical bisimulation.
\end{proof}
}{}

\begin{corollary}
  \label{cor:CLBNisABn}
  The relation $\clbn_2$ is an applicative bisimulation.
\end{corollary}

\arxiv{
\begin{lemma}
  Applicative bisimilarity is an equivalence relation.
\end{lemma}

\begin{proof}
  To show reflexivity, observe that $Id$ is an applicative
  bisimulation, so $M \ABn M$ for all $M$.

  To show symmetry, observe that if $M \ABn N$, then $M \R N$ for some
  applicative bisimulation.
  Then $\R^{op}\ \subseteq\ \ABn$ is also an applicative bisimulation,
  so $N \ABn M$.

  To show transitivity, observe that if $L \ABn M$ and $M \ABn N$,
  then $L \R M$ and $M \R' N$ for some applicative bisimulations $\R,
  \R'\ \subseteq\ \ABn$.
  Then $L \Ra \lambda x.L'$ iff $M \Ra \lambda x.M'$ such that
  $L'[V/x] \R M'[V/x]$ for all $V$, and $M \Ra \lambda x.M'$ iff $N
  \Ra \lambda x.N'$ such that $M'[V/x] \R' N'[V/x]$ for all $V$.
  Thus, $L \Ra \lambda x.L'$ iff $N \Ra \lambda x.N'$ such that
  $L'[V/x] \R\R' N'[V/x]$ for all $V$.
  The converse follows in an identical fashion.
  It is thus clear that $\R\R'$ is an applicative bisimulation, so $L
  \ABn N$ as desired.
\end{proof}

\begin{lemma}
  \label{lemma:bsra}
  We have the following containment of relations: $\Ra\ \subseteq\
  \ABn$.
\end{lemma}

\begin{proof}
  Assume $M \Ra N$ and that $M \Ra \lambda x.P$.
  Then by the determinacy of the call-by-name semantics, $N \Ra
  \lambda x.P$ and for all $V$, $P[V/x] \Ra P[V/x]$ since $\Ra$ is a
  reflexive relation. The converse follows identically.
  Thus, $\Ra$ is an applicative bisimulation, so $\Ra\ \subseteq\
  \ABn$.
\end{proof}
}{}

Although the big-step formulation prevents every applicative
bisimulation from being seen as a CLB via the mapping ${\R} \mapsto
({\Id}, {\R})$, we at the very least have that every applicative
bisimulation is, in a certain sense, contained in a CLB: \RAK{If I
  have time, show that every AB can be seen as a CLB up-to reduction
  via the aforementioned mapping.}

\begin{prop}
  \label{prop:ABnisCLBn}
  The coupled relation $({\Id}, {\ABn})$ is a CLB.
\end{prop}

\arxiv{
  \begin{proof}
    By Lemma \ref{lemma:bsra} and the fact that ${\ra} \subseteq
    {\Ra}$, ${\ra} \subseteq {\ABn}$.
    So if $M \ra M'$ and $M \ABn N$, then since $N \Ra N$, by
    transitivity, $M' \ABn N$ as desired.
    If $M = \lambda x.P$, then since $M \ABn N$, $N \Ra \lambda x.Q$
    such that $P[W/x] \ABn Q[W/x]$ for all $W$ (recall that $Id^\star
    = Id$).
    Thus, $({\Id}, {\ABn})$ is a CLB.
  \end{proof}
}{
}

\begin{corollary}
  \label{cor:stringeqs1}
  Applicative bisimilarity, coupled logical bisimilarity, and
  contextual equivalences coincide, i.e., ${\ABn} = {\clbn_2} =
  {\ecen} = {\cen} = {\clbn_1}$.
  Thus, applicative bisimilarity is a congruence.
\end{corollary}

\arxiv{
  \begin{proof}
    We have that ${\ABn} = {\clbn_2}$ by double inclusion via
    Proposition \ref{prop:ABnisCLBn} and Corollary
    \ref{cor:CLBNisABn}.
    The other equalities follow from Corollary
    \ref{cor:clbn1-eq-clbn2}.
  \end{proof}
}{}
\subsection{Logical Bisimulation}

\textcite{Sangiorgi2007} introduced the notion of \textit{logical
  bisimulation}, a notion we show to be subsumed by coupled logical
bisimulation.

\begin{definition}
  A relation $\R\ \subseteq\ \Ld\times\Ld$ is called a \textbf{logical
    bisimulation} if whenever $M \R N$:
  \begin{enumerate}
  \item if $M \ra M'$, then $N \Ra N'$ and $M' \R N'$;
  \item if $M = \lambda x.M'$, then $N \Ra \lambda x.N'$ and for all
    $X \R^\star Y$, $M'[X/x] \R N'[Y/x]$;
  \item the converses for $N$.
  \end{enumerate}
  The union of all logical bisimulations is called \textbf{logical
    bisimilarity} and is denoted $\lbn$.
\end{definition}

As one would expect, we have the following proposition, the proof of
which may be found in \cite[Corollary 1 and Lemma 4]{Sangiorgi2007}:

\begin{prop}[\cite{Sangiorgi2007}]
  \label{prop:lbisbisimandcong}
  Logical bisimilarity is the largest logical bisimulation and is a
  congruence relation.
\end{prop}

The following proposition follows from a straightforward check of the
definitions and tells us that the notion of logical bisimulation is
subsumed by that of CLB:

\begin{prop}
  \label{prop:lbiffclb}
  A relation $\R$ is a logical bisimulation if and only if $(\R, \R)$
  is a CLB.
\end{prop}

\begin{corollary}
  Logical bisimilarity, coupled logical bisimilarities, applicative
  bisimilarity, and contextual equivalences coincide, i.e., ${\lbn} =
  {\ABn} = {\clbn_2} = {\cen} = {\ecen} = {\clbn_1}$.
\end{corollary}
\subsection{Applicative and Logical Bisimulations}

}{

We now compare CLB with AB and LB.
We first recall the big-step version of applicative bisimulation, as
originally presented by Abramsky \cite{Abramsky1990}:

\begin{definition}[Applicative bisimulation]
  A relation $\R\ \subseteq \Ld\times\Ld$ is called an
  \textbf{applicative bisimulation} if $M\R N$ implies that whenever
  $M \Ra \lambda x.P$, $N \Ra \lambda x.Q$ for some $Q$, and $P[W/x]
  \R Q[W/x]$ for all $W \in \Ld$, and conversely for $N$.
  We call the union of all applicative bisimulations, written $\ABn$,
  \textbf{applicative bisimilarity}.
\end{definition}

It is not hard to show that applicative bisimilarity is itself an
applicative bisimulation.

\begin{prop}
  \label{prop:CLBAB}
  If $({\R'}, {\R})$ is a CLB for some ${\R'} \subseteq {\R} \cap
  {\Id}$, then $\R$ is an applicative bisimulation.
\end{prop}

\begin{corollary}
  \label{cor:CLBNisABn}
  The relation $\clbn_2$ is an applicative bisimulation.
\end{corollary}

\begin{prop}
  \label{prop:ABnisCLBn}
  The coupled relation $({\Id}, {\ABn})$ is a CLB.
\end{prop}


As in the cbn case,  \textit{logical
  bisimulation}~\textcite{Sangiorgi2007} is subsumed by coupled
logical bisimulation.

\begin{definition}[Logical bisimulation]
  A relation $\R\ \subseteq\ \Ld\times\Ld$ is called a \textbf{logical
    bisimulation} if whenever $M \R N$:
  \begin{enumerate}
  \item if $M \ra M'$, then $N \Ra N'$ and $M' \R N'$;
  \item if $M = \lambda x.M'$, then $N \Ra \lambda x.N'$ and for all
    $X \R^\star Y$, $M'[X/x] \R N'[Y/x]$;
  \item the converses for $N$.
  \end{enumerate}
  The union of all logical bisimulations is called \textbf{logical
    bisimilarity} and is denoted $\lbn$.
\end{definition}

As one would expect, we have the following proposition, the proof of
which may be found in \cite[Corollary 1 and Lemma 4]{Sangiorgi2007}:

\begin{prop}[\cite{Sangiorgi2007}]
  \label{prop:lbisbisimandcong}
  Logical bisimilarity is the largest logical bisimulation and is a
  congruence relation.
\end{prop}

\begin{prop}
  \label{prop:lbiffclb}
  A relation $\R$ is a logical bisimulation if and only if $(\R, \R)$
  is a CLB.
\end{prop}

\begin{corollary}
  Logical bisimilarity, coupled logical bisimilarities, applicative
  bisimilarity, and contextual equivalences coincide, i.e., ${\lbn} =
  {\ABn} = {\clbn_2} = {\cen} = {\ecen} = {\clbn_1}$.
\end{corollary}

}

\section{CLB in the Call-by-value $\lambda$-calculus}
\label{s:cbv}

We now move to the study of call-by-value.

\begin{definition}
  The \textbf{call-by-value} $\lambda$-calculus is defined by the
  following reduction rules:
  \begin{equation*}
    \frac{N\ra N'}{MN\ra MN'}\hspace{3em}
    \frac{M\ra M'\quad V\in\V}{MV\ra M'V}\hspace{3em}
    \frac{V\in\V}{(\lambda x.P)V \ra P[V/x]},
  \end{equation*}
  where we take the set $\V$ of \textbf{values} to be the set of all
  abstractions $\lambda x.P \in \Lambda$ and all variables $x$.
\end{definition}

Since our theory is centered around closed terms, we restrict the set
$\V$ to the set of all $\lambda x.P \in \Ld$ throughout our
development.
We let the relations $\Ra$, $\conv$, and $\diver$, and the predicates
$\conv$ and $\diver$ be as before, except using call-by-value
reduction instead of call-by-name reduction.

\begin{definition}
  We say that two terms $M, N \in \Ld$ are \textbf{contextually
    equivalent}, written $M \cev N$, if for all contexts $C$,
  $C[M]\conv$ if and only if $C[N]\conv$.
\end{definition}

\begin{definition}
  A \textbf{call-by-value evaluation context} $\EC$ is given by the
  following grammar,
  \[ \EC := [\cdot] \mid M\EC \mid \EC V, \] where $V$ ranges over
  $\V$ and $M$ ranges over $\Ld$.
  Two terms $M$ and $N$ are \textbf{evaluation-contextually
    equivalent}, written $M \ecev N$, if for all evaluation contexts
  $\EC$, $\EC[M]\conv$ if and only if $\EC[N]\conv$.
\end{definition}

\begin{definition}
  If $\R, \R'\ \subseteq \Ld\times\Ld$ are relations, then $\R$'s
  \textbf{evaluation-contextual closure under $\R'$},
  $\eccv{\R}{\R'}$, is the least relation closed forward under the
  following rules: \[
  \frac{X \R Y}{X \eccv{\R}{\R'} Y}\quad
  \frac{M \R' N\quad X \eccv{\R}{\R'} Y}{MX \eccv{\R}{\R'} NY}\quad
  \frac{V \vr{\R'} W\quad X \eccv{\R}{\R'} Y}{XV \eccv{\R}{\R'} YW}.
  \]
\end{definition}

\begin{definition}
  If $\R$ is a coupled relation, then let its \textbf{contextual closure}
  $\R^V$ be given by
  \[\R^V = \left(\R_1^\star, 
    \eccv{\R_2}{\R_1^\star} \cup
    \R_1^\star \right).\]
\end{definition}

\subsection{Proving the Context Lemma for call-by-value}
\label{s:context:lemma:cbv}

As for the call-by-name $\lambda$-calculus, we have a Milner-style
context lemma. Although allusions to a proof exist in the literature,
e.g., in a footnote in~\cite{Milner1990}, the authors have been
unable to find a published proof. 
We present ours
  below.

\begin{theorem}
  \label{thm:cev-is-ecev}
  We have the following equality of relations: ${\cev} = {\ecev}$.
\end{theorem}
\ifarxivversion
\begin{lemma}
  \label{lemma:ecevevctxtcong}
  The relation $\ecev$ is closed forward under the following rules:
  \[ \frac{A \ecev B \quad M\in\Ld}{MA \ecev MB}\quad\quad \frac{A
    \ecev B \quad V\in\V}{AV \ecev BV}.\]
\end{lemma}

\begin{proof}
  We consider the first rule.  Let $\EC$ is an evaluation context, we
  prove that $\EC[M A] \conv$ if and only if $\EC[M B] \conv$, which
  is equivalent to $\EC'[A] \conv$ if and only if $\EC'[B] \conv$ with
  $\EC' = \EC(M[\cdot])$.  The latter holds since $A \ecev B$.  The
  second rule follows in an identical manner with
  $\EC'=\EC([\cdot]V)$.
\end{proof}


\begin{corollary}
  The relation $\ecev$ is closed forward under the following rule:
  \[ \frac{M \ecev N \quad V \vr{\ecev} W}{MV \ecev NW}. \]
\end{corollary}

\begin{proof}
  Assume $M \ecev N$ and $V \vr{\ecev} W$.
  Then by Lemma \ref{lemma:ecevevctxtcong}, we have $MV \ecev NV$ and
  $NV \ecev NW$.
  Then by transitivity of $\ecev$, we get $MV \ecev NW$.
\end{proof}
\fi

\begin{lemma}
  \label{lemma:RainECEV}
  If $M \Ra N$, then $M \ecev N$.
\end{lemma}

\begin{proof}
  If $\EC$ is an evaluation-context, then $\EC[M] \Ra \EC[N]$, and so
  by determinism of $\ra$, we know that $\EC[M]$ converges if and only
  if $\EC[N]$ does.
\end{proof}

\begin{prop}
  \label{prop:cmconviffcvconv}
  If $M \convr V$ then $M \cev V$.
\end{prop}
\begin{proof}[\ifarxivversion Proof\fi]
  \newcommand{\rap}{\longtwoheadrightarrow}%
  Consider the relation:
  \[ {\R} = \{(C\fl{M}, C\fl{V}) \mid C \mbox{ is a $n$-holed context,
  } n \in \N \}.\]%
  We show that whenever $P \R Q$, if $P \in \V$ then $Q \conv$, and if
  $P \ra P'$, then $P' \Ra \R \Longleftarrow Q$.  We show the same
  property for $\R\op$ which is enough to conclude by determinism of
  $\ra$.
  
  Intuitively, if on the left $M\ra M'$ then we complete with $M'\Ra
  V$ and we proceed by taking a $V$ off the $C$.  When $C\fl{M}$ is
  moving on the left without touching $M$ (maybe duplicating it), we
  can do the same on the right. When $C\fl{V}$ moves on the right, it
  does the same on the left, but before we might have to reduce some
  $M$ blocking a reduction into $V$. We proceed with an induction on
  $C$.
  \ifarxivversion
  
  The cases when $C\fl{M}\in\V$ or $C\fl{V}\in\V$ are immediate since
  then $C$ is canonical.

  Suppose now $C\fl{M}\ra$ or $C\fl{V}\ra$.  If $C$ is canonical then
  $C\fl{M}=M \ra M' \Ra V = C\fl{V}$ so we proceed with the $0$-holed
  context $V$. We suppose now that $C$ is not canonical ($C=C_1C_2$)
  and we prove by induction on $C$ that there exists $C'$ such that
  $C\fl{V} \ra C'\fl{V}$ and $C\fl{M} \rap C'\fl{M}$ where $\rap$ is
  the transitive closure of $\ra$.  Suppose first that $C_1$ or $C_2$
  is not canonical.
  \begin{itemize}
  \item If $C_2$ is not canonical, then by induction $C_2\fl{V} \ra
    C_2'\fl{V}$ and $C_2\fl{M} \rap C_2'\fl{M}$ and thus $C_1C_2\fl{V}
    \ra C_1C_2'\fl{V}$ and $C_1C_2\fl{M} \rap C_1C_2'\fl{M}$. (Context
    $C'$ is then $C_1C_2'$.)
  \item   
    If $C_1$ is not canonical and $C_2$ is but $C_2\neq[\cdot]$ then
    by induction $C_1\fl{V} \ra C_1'\fl{V}$ and $C_1\fl{M} \rap
    C_1'\fl{M}$ and thus $C_1C_2\fl{V} \ra C_1'C_2\fl{V}$ and
    $C_1C_2\fl{M} \rap C_1'C_2\fl{M}$ because $C_2\fl{M}$ is a
    value. (Context $C'$ is then $C_1'C_2$.)
  \item 
    If $C_1$ is not canonical and $C_2 = [\cdot]$ is then by induction
    $C_1\fl{V} \ra C_1'\fl{V}$ and $C_1\fl{M} \rap C_1'\fl{M}$ and
    thus $C_1\fl{V}V \ra C_1'\fl{V}V$ and $C_1\fl{M}M \Ra C_1\fl{M}V
    \rap C_1'\fl{M}V$. ($C'$ is then $(C_1')V$.)
  \end{itemize}
  We now handle the cases where both $C_1$ and $C_2$ are
  canonical. Let $R$ be such that $V = \lambda x.R$.
  \begin{itemize}
  \item $C_1 = [\cdot] = C_2 $: then $V V\ra R[V / x]$ and $M M \Ra M
    V \Ra V V \ra R[V / x]$. ($C'$ is the $0$-holed context $R[V /
    x]$.)
  \item $C_1 = [\cdot] \neq C_2$: then $V C_2\fl{V} \ra R[C_2\fl{V}/x]$
    and $M C_2\fl{M} \Ra V C_2\fl{M} \ra R[C_2\fl{M}/x]$ since $C_2\fl{M} \in\V$.
    (Then $C'$ is $R[C_2/x]$.)
  \item $C_1=\lambda y.D_1$ and $C_2 = [\cdot]$: then $(\lambda
    y.D_1\fl{V})V \ra D_1[V/y]\fl{V}$ and $(\lambda y.D_1\fl{M})M \Ra
    (\lambda y.D_1\fl{M})V \ra D_1[V/y]\fl{M}$.  (Then $C'$ is
    $D_1[V/y]$.)
  \item $C_1=\lambda y.D_1$ and $C_2 \neq [\cdot]$: then $(\lambda
    y.D_1\fl{V})C_2\fl{V} \ra D_1[C_2\fl{V}/y]\fl{V}$ and $(\lambda
    y.D_1\fl{M})C_2\fl{M}$ $\ra D_1[C_2\fl{M}/y]\fl{M}$. (Then
    $C'$ is $D_1[C_2/y]$.)
  \end{itemize}
  This case analysis shows us that if $P \mathrel{ \mathcal{S}} Q$ and
  $P \ra P'$ then $P' \Ra \mathrel{ \mathcal{S}} \Longleftarrow Q$ for
  both $\mathcal{S} \in \{{\R},{\R\op}\}$.  From this, and the clause
  about $P\in\V$, we can easily prove that $P \mathrel{ \mathcal{R}}
  Q$ implies ($P\conv$ iff $Q\conv$).
  \fi
\end{proof}

\begin{prop}
  \label{prop:jmd}
  If $V,W\in\V$ and $V \ecev W$ then $V \cev W$.
\end{prop}

\begin{proof}
  We prove $C[V] \conv$ iff $C[W] \conv$.  Consider the evaluation
  context $\EC = (\lambda x.C[x])[\cdot]$ yielding $\EC[V] \ra C[V]$
  and $\EC[W] \ra C[W]$.  Since $V \ecev W$, $\EC[V] \conv$ iff
  $\EC[W] \conv$; hence $C[V] \conv$ iff $C[W] \conv$.
\end{proof}

The above proposition can be strengthened as follows:

\begin{prop}
  \label{prop:jmdrak}
  If $M \conv$ and $N \conv$ then $M \ecev N$ implies $M \cev N$.
\end{prop}

\begin{proof}
  Let $V$ and $W$ such that $M \conv V$, $N \conv W$.  By Proposition
  \ref{prop:cmconviffcvconv}, $M \cev V$ and $N \cev W$. Then, if $M
  \cev N$, then $V \ecev W$, then $V \cev W$ by
  Proposition~\ref{prop:jmd} and finally $M \cev N$.
\end{proof}

\begin{prop}
  \label{prop:divercev}
  If $M\diver$ and $N\diver$, then $M \cev N$.
\end{prop}

\begin{proof}
  Consider the symmetric relation:
  \[ {\R} = \{(C\fl{M}, C\fl{N}) \mid C \mbox{ is a $n$-holed context}
  \} \cup \{(\EC_1[M'],\EC_2[N'] \mid M'\diver \mbox{ and }
  N'\diver)\}.\]%
  Suppose $P \R Q$.  Trivially, if $P \in \V$ then $Q\in V$. Remains
  to prove that $P \ra P'$, then $Q \ra Q'$ and $P' \R Q'$ for some
  $Q'$.  The second part of the relation is trivial.  Regarding the
  first part, if $M$ appears in evaluation position ($P=\EC_1[M]$)
  then so does $N$ ($Q=\EC_2[N]$), the pair progressing to the second
  part of the relation.  If $M$ does not, then $C\fl{M} \ra C'\fl{M}$
  and $C\fl{N} \ra C'\fl{N}$.
\end{proof}



\begin{corollary}
  \label{cor:cevecevcoincide}
  $M \cev N$ if and only if $M \ecev N$.
\end{corollary}

\begin{proof}
  Clearly $M \cev N$ implies $M \ecev N$.  Considering $\EC=[\cdot]$,
  either $M \conv$ and $N\conv$ or $M \diver$ and $N\diver$ and using
  Propositions~\ref{prop:jmdrak} and~\ref{prop:divercev} we derive the
  other implication.
\end{proof}

\ifarxivversion
\begin{lemma}
  \label{lemma:eccvreduxcases}
  If $M \eccv{\R}{\R'} N$ and $M \ra M'$, then one of the following
  cases holds
  \begin{enumerate}
  \item $M \R N$;
  \item $M = \vec{E_M}[\alpha]$, $N = \vec{E_N}[\beta]$, $\alpha \R
    \beta$, $\|\vec{E_M}\| = \|\vec{E_N}\| = n$, $E_{Mn} = [\cdot]V$,
    $E_{Nn} = [\cdot]W$, $V \vr{\R'} W$, and $\alpha \ra \alpha'$;
  \item $M = \vec{E_M}[\alpha]$, $N = \vec{E_N}[\beta]$, $\alpha \R
    \beta$, $\|\vec{E_M}\| = \|\vec{E_N}\| = n$, $E_{Mn} = [\cdot]V$,
    $E_{Nn} = [\cdot]W$, $V \vr{\R'} W$, $\alpha = \lambda x.\alpha'$,
    and $M \ra E_{M1}[\cdots[E_{M(n-1)}[\alpha'[V/x]]]\cdots]$;
  \item $M = \vec{E_M}[\alpha]$, $N = \vec{E_N}[\beta]$, $\alpha \R
    \beta$, $\|\vec{E_M}\| = \|\vec{E_N}\| = n$, $E_{Mn} = X[\cdot]$,
    $E_{Nn} = Y[\cdot]$, $X \R' Y$, and $\alpha \ra \alpha'$;
  \item $M = \vec{E_M}[\alpha]$, $N = \vec{E_N}[\beta]$, $\alpha \R
    \beta$, $\|\vec{E_M}\| = \|\vec{E_N}\| = n$, $E_{Mn} = X[\cdot]$,
    $E_{Nn} = Y[\cdot]$, $X \R' Y$, $\alpha \in \V$, and $X \ra X'$;
  \item $M = \vec{E_M}[\alpha]$, $N = \vec{E_N}[\beta]$, $\alpha \R
    \beta$, $\|\vec{E_M}\| = \|\vec{E_N}\| = n$, $E_{Mn} = X[\cdot]$,
    $E_{Nn} = Y[\cdot]$, $X \R' Y$, $X = \lambda x.X', \alpha =
    \lambda y.\alpha'$, and $M \ra
    E_{M1}[\cdots[E_{M(n-1)}[X'[\alpha/x]]]\cdots]$.
  \end{enumerate}
  Moreover, where $E_{Mi}$ and $E_{Ni}$ are such that $\vec{E_M} =
  E_{M1},\dotsc,E_{Mn}$ and $\vec{E_N} = E_{N1},\dotsc,E_{Nn}$,
  whenever $E_{Mi} = [\cdot]V_i$ then $E_{Ni} = [\cdot]W_i$ for some
  $W_i$ with $V_i \vr{\R'} W_i$ and conversely, and whenever $E_{Mi} =
  X_i[\cdot]$ then $E_{Ni} = Y_i[\cdot]$ for some $Y_i$ with $X_i \R'
  Y_i$ and conversely.
\end{lemma}
\fi


\ifarxivversion
\begin{lemma}
  \label{lemma:eccvcombine}
  If $\R, \R'\ \subseteq \Ld\times\Ld$ are relations, $\vec{E} =
  E_1,\dotsc,E_n$ and $\vec{F} = F_1,\dotsc,F_n$ are lists of
  experiments such that whenever $E_i = [\cdot]V_i$ then $F_i =
  [\cdot]W_i$ for some $W_i$ with $V_i \vr{\R'} W_i$ and conversely,
  and whenever $E_i = M_i[\cdot]$ then $F_i = N_i[\cdot]$ for some
  $N_i$ with $M_i \R' N_i$ and conversely, then for all $\alpha \R
  \beta$, $\vec{E}[\alpha] \eccv{\R}{\R'} \vec{F}[\beta]$.
\end{lemma}

\begin{proof}
  By induction on $n$.
  The case of $n = 0$ is trivial, so we assume true for some $n-1$,
  and let $\vec{E} = E_1,\vec{E'}$ and $\vec{F} = F_1,\vec{F'}$ be two
  lists of experiments of length $n$ satisfying the hypotheses.
  Then by the induction hypothesis, since $\vec{E'}$ and $\vec{F'}$
  are appropriately related lists of length $n-1$, for all $\alpha \R
  \beta$, $\vec{E'}[\alpha] \eccv{\R}{\R'} \vec{F'}[\beta]$.
  If $E_1 = [\cdot]V_1$, then $F_1 = [\cdot]W_1$ for $V_1 \vr{\R'}
  W_1$, and so by definition of $\eccv{\R}{\R'}$, we get
  $\vec{E}[\alpha] = \vec{E'}[\alpha]V_1 \eccv{\R}{\R'}
  \vec{F'}[\beta]W_1 = \vec{F}[\beta]$ for all $\alpha \R \beta$ as
  desired.
  The case of $E_1 = M_1[\cdot]$ follows in a similar manner.
  We thus conclude the lemma by induction.
\end{proof}

\begin{lemma}
  \label{lemma:nestrvineccv}
  If $\vec E = E_1(E_2(\cdots(E_n)\cdots))$ and $\vec F =
  F_1(F_2(\cdots(F_n)\cdots))$ and $E_i \R_1 F_i$ for $1 \leq i \leq
  n$ and $\alpha \R^V_2 \beta$, then $\vec E[\alpha] \R^V_2 \vec
  F[\beta]$.
\end{lemma}

\begin{proof}
  We proceed by case analysis on why $\alpha \R^V_2 \beta$.
  If the relation holds because $\alpha \eccv{\R_2}{\R_1^\star}
  \beta$, then $\alpha = \vec{A}[\xi]$ and $\beta = \vec{B}[\psi]$
  where $\vec{A}$ and $\vec{B}$ are lists of equal length of $A_i$ and
  $B_i$ such that $A_i \R_1^\star B_i$, and $\xi \R_2 \psi$.
  Let $\vec{G_\alpha} = \vec{E},\vec{A}$ and $\vec{G_\beta} =
  \vec{F},\vec{B}$.
  Then we get $\vec{E}[\alpha] = \vec{G_\alpha}[\xi]
  \eccv{\R_2}{\R_1^\star} \vec{G_\beta}[\psi] = \vec{F}[\beta]$ as
  desired.

  If the relation holds because of $\R_2$, then we fall into the
  previous case.

  If the relation holds because of $\R_1^\star$, then $\vec{E}[\alpha]
  \R_1^\star \vec{F}[\beta]$ and so we're done since $\R_1^\star\
  \subseteq\ \R^V_2$.
\end{proof}
\fi

\subsection{Coupled Logical Bisimulation}

The definition of coupled logical bisimulation for the call-by-value
calculus differs 
from that for the call-by-name calculus,
viz., the second clause.
The additional requirement in the second clause plays a central role
in showing that coupled logical bisimilarity is a CLB.


\begin{definition}
  A coupled relation $\R$ is a \textbf{coupled logical bisimulation}
  (CLB) if whenever $M \R_2 N$, we have:
  \begin{enumerate}
  \item if $M \ra M'$, then there exists an $N'$ such that $N \Ra N'$
    and $M' \R_2 N'$;
  \item if $M = \lambda x.M'$, then $N \Ra \lambda x.N'$, $\lambda
    x.M' \R_1 \lambda x.N'$, and for all $P, Q \in \Ld$ such that $P
    \vr{\R_1^\star} Q$, we have $M' [P/x]\R_2N'[Q/x]$;
  \item the converses of the previous two conditions for $N$.
  \end{enumerate}
  Coupled logical bisimilarity, written ${\clbv} = (\clbv_1,
  \clbv_2)$, is the pairwise union of all CLBs.
\end{definition}

As in the call-by-name case, CLBs for the call-by-value
$\lambda$-calculus have a continuous progression:

\begin{definition}
  \label{def:progressescbv}
  Given pairs of relations $\R$ and $\SR$, we say
  \textbf{$\R$ progresses to $\SR$}, written $\R\ \prog\
  \SR$, if whenever $M \R_2 N$, then:
  \begin{enumerate}
  \item whenever $M \ra M'$ then $N \Ra N'$ and $M' \SR_2 N'$;
  \item whenever $M = \lambda x.P$ then $N \Ra \lambda x.Q$ such that
    $\lambda x.P \SR_1 \lambda x.Q$ and for all $X \vr{\R_1^\star} Y$,
    $P[X/x] \SR_2 Q[Y/x]$;
  \item the converses of the previous two conditions for $N$.
  \end{enumerate}
\end{definition}

\begin{prop}
  \label{prop:equalthenclb}
  A pair of relations $\R$ is a CLB if and only if ${\R} \prog {\R}$.
  Thus, $\prog$ is a progression for coupled logical bisimulations.
\end{prop}

\begin{prop}
  \label{prop:clbvsympa}
  The relation $\prog$ is continuous.
\end{prop}

\begin{proof}
  \newcommand{\nr}{\mathrel{(\nu{\R})}}
  \newcommand{\ns}{\mathrel{(\nu{\SR})}}
  \newcommand{\rn}{\mathrel{({\R}_n)}}
  \newcommand{\sn}{\mathrel{({\SR}_n)}}

  \arxiv{The proof is identical to that of Proposition
    \ref{prop:clbnsympa} apart from the case $M = \lambda x.M'$, which
    now reads as:

    If $M = \lambda x.M'$, then, since ${\R_n} \prog {\SR_n}$, $N \Ra
    \lambda x.N'$ such that $\lambda x.M' \sn_1 \lambda x.N'$, and so
    since ${\SR_n} \subseteq \nu{\SR}$, we get that $\lambda x.M'
    \ns_1 \lambda x.N'$.
    Moreover, $M'[X/x] \sn_2 N'[Y/x]$ for all $X \rn_1^\star Y$, so
    $M'[X/x] \ns_2 N'[Y/x]$ and we're done.
  }{ Similar to Proposition
    \ref{prop:clbnsympa},  the clause $\lambda x.M'
    \ns_1 \lambda x.N'$ being handled in the same way.
  }
\end{proof}

\begin{prop}
  \label{prop:clbv-prog-union}
  \newcommand{\Ri}{\mathrel{({\R_i})}}
  \newcommand{\SRi}{\mathrel{({\SR_i})}}

  If $\left\{{\R_i}\right\}_{i\in I}$ and
  $\left\{{\SR_i}\right\}_{i\in I}$ are two families of paired
  relations such that ${\R_i} \prog {\SR_i}$ for all $i \in I$, then
  $\left(\bigcap_{i\in I} {\Ri_1}, \bigcup_{i\in I} {\Ri_2}\right)
  \prog \bigcup_{i\in I} {\SR_i}$.
\end{prop}

\begin{proof}
  \renewcommand{\SS}{\mathrel{\left(\bigcup_{i\in I} {\SR_i}\right)}}
  \newcommand{\UU}{\mathrel{\U}}
  \newcommand{\Ri}{\mathrel{({\R_i})}}
  \newcommand{\SRi}{\mathrel{({\SR_i})}}

  Let ${\UU} = (\bigcap_{i\in I} \Ri_1, \bigcup_{i\in I} \Ri_2)$, and
  assume $M \UU_2 N$.
  Then there exists some $\R_i$ such that $M \Ri_2 N$.

  If $M \ra M'$, then $N \Ra N'$ \arxiv{such that}{with} $M' \SRi_2 N'$, and by
  inclusion, we get that $M' \SS_2 N'$.

  If $M = \lambda x.M'$ and $X \UU_1^\star Y$, then since ${\UU_1}
  \subseteq {\Ri_1}$, we have $X \Ri_1^\star Y$ by the monotonicity of
  contextual closure.
  Then since ${\R_i} \prog {\SR_i}$, we have $N \Ra \lambda x.N'$,
  $\lambda x.M' \SRi_1 \lambda x.N'$, and $M'[X/x] \SRi_2 N'[Y/x]$.
  Then by inclusion, we get that $\lambda x.M' \SS_1 \lambda x.N'$ and
  $M'[X/x] \SS_2 N'[Y/x]$ as desired.

  The symmetric cases for $N$ follow symmetrically, and we are done.
\end{proof}

\arxiv{
\begin{lemma}
  \label{lemma:clbvconv}
  If $\R$ is a CLB and $M\R_2N$, then there exists a $V$ such that
  $M\conv V$ if and only if there exists a $W$ such that $N \conv W$.
\end{lemma}

\begin{proof}
  Immediate by the definition of CLBV.
\end{proof}

\begin{lemma}
  \label{lemma:IdRaCLBV}
  The coupled relation $(\vr{\Id}, \Ra)$ is a CLB.
\end{lemma}

\begin{proof}
  Assume $M \Ra M'$ and that $M \ra M''$.
  If $M = M'$, then $M' \Ra M''$ also, and $M'' \Ra M''$.
  If $M \neq M'$, then $M'' \Ra M''$, and $M' \Ra M''$.
  Conversely, if $M'' \Ra M'$, then $M \Ra M$ and by transitivity of
  $\Ra$, $M \Ra M'$.

  Now assume $M = \lambda x.P$, then $M'' = \lambda x.P$, so $M'' \Ra
  \lambda x.P$ by reflexivity, and $\lambda x.P \vr{\Id} \lambda x.P$.
  Clearly for all $V \vr{\Id} W$, $P[V/x] \Ra P[W/x]$ by reflexivity,
  since $P[V/x] = P[W/x]$.
  Conversely, if $M' = \lambda x.P$, then $M \Ra \lambda x.P$ and the
  same argument applies.
\end{proof}
}{}

As in Section~\ref{s:cbn}, we study the up-to context technique, which
allows us to deduce congruence.
\begin{definition}
  We call \textbf{up-to context} the up-to technique given by ${\R}
  \mapsto {\R}^V$.
  We say a coupled relation $\R$ is a \textbf{CLB up-to context} if
  ${\R} \prog {\R^V}$.
\end{definition}

\begin{lemma}
  \label{lemma:rstarclbv}
  If ${\R} \prog {\SR}$ and ${\R} \subseteq {\SR}$, then $({\R_1},
  {\R_1^\star}) \prog {\SR^V}$.
\end{lemma}

\ifarxivversion
\begin{proof}
  Suppose $M \R_1^\star N$.
  It is sufficient to show that: \begin{inparaenum}[\itshape (i)
    \upshape]
  \item if $M \ra M'$, then there exists an $N'$ such that $N \Ra N'$
    and $M' \SR^V_2 N'$, and symmetrically if $N \ra N'$;
  \item if $M = \lambda x.M'$, then there exists an $N'$ such that $N
    \Ra \lambda x.N'$, $\lambda x.M' \SR^V_1 \lambda x.N'$, and for
    all $X \vr{\R_1^\star} Y$, $M'[X/x] \SR^V_2 N'[Y/x]$, and
    symmetrically if $N = \lambda x.N'$.
  \end{inparaenum}

  Let $\wt{M} \R_1 \wt{N}$ such that $M = C[\wt{M}]$ and $N =
  C[\wt{N}]$ and assume first that $M \ra M'$.
  Then we proceed by induction on $C$, and observe that it is
  sufficient to show either of $M' \SR_1^\star N'$ or $M'
  \eccv{\SR_2}{\SR_1^\star} N'$.

  The case $C = [\cdot]$ is trivial since it implies $M \R_1 N$, and
  we have as hypothesis that ${\R} \prog {\SR}$.

  The cases $C = x$ and $C = \lambda x.C'$ does not arise since
  $C[\wt{M}] \not\ra$.

  Finally, we consider the case $C = C_1C_2$ and let $\wt{M_1},
  \wt{M_2}$ be such that $C[\wt{M}] = C_1[\wt{M_1}] C_2[\wt{M_2}]$ and
  similarly for $\wt{N}$.
  By monotonicity of contextual closure, we observe that $A \R_1^\star
  B$ implies $A \SR_1^\star B$ for all $A, B$ since ${\R} \subseteq
  {\SR}$.
  The reduction $M \ra M'$ is due to one of the following mutually
  exclusive subcases:
  \begin{enumerate}
  \item $C_2[\wt{M_2}] \ra M_2'$, so $M \ra C_1[\wt{M_1}]M_2'$;
  \item $C_2[\wt{M_2}] \in \V$ and $C_1[\wt{M_1}] \ra M_1'$, so $M \ra
    M_1'C_2[\wt{M_2}]$;
  \item $C_1[\wt{M_1}] = \lambda x.P, C_2[\wt{M_2}] \in \V$, so $M \ra
    P[C_2[\wt{M_2}]/x]$.
  \end{enumerate}

  In the first case, since $C_2[\wt{M_2}] \R_1^\star C_2[\wt{N_2}]$,
  by induction hypothesis, $C_2[\wt{N_2}] \Ra N_2'$ for some $N_2'$
  such that $M_2' \SR_2^V N_2'$, and $N \Ra N' := C_1[\wt{N_1}]N_2'$.
  If $M_2' \SR_2^V N_2'$ holds because of $\SR_1^\star$, then Lemma
  \ref{lemma:rccomb} gives us $C_1[\wt{M_1}] M_2' \SR_1^\star
  C_1[\wt{N_1}]N_2'$ and we're done.
  If it holds because of $\SR_2$, then we're done because
  $C_1[\wt{M_1}] M_2' \eccv{\SR_2}{\SR_1^\star} C_1[\wt{N_1}]N_2'$.
  Finally, if it holds because of $\eccv{\SR_2}{\SR_1^\star}$, then
  $C_1[\wt{M_1}] \SR_1^\star C_1[\wt{N_1}]$ implies $C_1[\wt{M_1}]
  M_2' \eccv{\SR_2}{\SR_1^\star} C_1[\wt{N_1}]N_2'$ as desired.

  In the second case, $C_2[\wt{M_2}] \in \V$ implies $C_2$ is a
  canonical context.
  If $C_2 = [\cdot]$, then $\wt{M_2} = M_2$ is a single element list,
  and since $M_2 \R_1 N_2$ and ${\R} \prog {\SR}$, we have $M_2 \in
  \V$ and $N_2 \Ra N_2' \in \V$ such that $M_2 \SR_1^\star N_2'$.
  Otherwise, $C_2 = \lambda x.C_2'$ and $C_2[\wt{N_2}] \in \V$
  trivially; in this case, let $N_2' = C_2[\wt{N_2}]$.
  In either case, $C_2[\wt{M_2}] \vr{\SR_1^\star} N_2'$.
  By the induction hypothesis, since $C_1[\wt{M_1}] \R_1^\star
  C_1[\wt{N_1}]$ and $C_1[\wt{M_1}] \ra M_1'$, $C_1[\wt{N_1}] \Ra
  N_1'$ such that $M_1' \SR^V_2 N_1'$; thus, $N \Ra N' := N_1'N_2'$.
  If $M_1' \SR^V_2 N_1'$ because of $\SR_1^\star$, then Lemma
  \ref{lemma:rccomb} gives us $M_1'C_2[\wt{M_2}] \SR_1^\star N_1'N_2'$
  and we're done.
  If the relation holds because of $\SR_2$, then $M_1'C_2[\wt{M_2}]
  \eccv{\SR_2}{\SR_1^\star} N_1' N_2'$ and again we're done.
  Finally, if the relation holds because of
  $\eccv{\SR_2}{\SR_1^\star}$, then $C_2[\wt{M_2}] \vr{\SR_1^\star}
  N_2'$ implies $M_1'C_2[\wt{M_2}] \eccv{\SR_2}{\SR_1^\star} N_1'N_2'$
  by Lemma \ref{lemma:eccvcombine} as desired.

  Finally, in the third case, the same argument as in the second case
  gives us that $C_2[\wt{N_2}] \Ra N_2' \in \V$, and similarly that
  $C_1[\wt{N_1}] \Ra \lambda x.Q \in \V$ with $P \R_1^\oclos Q$. Since
  $C_2[\wt{M_2}] \vr{\SR_1^\star} N_2'$ and $\fv(P) = \fv(Q) = \{x\}$,
  by Lemma \ref{lemma:starsubst}, we then conclude $P[C_2[\wt{M_2}]/x]
  \SR_1^\star Q[N_2'/x]$ and we're done.

  This exhausts all possible reductions in the case of $C = C_1C_2$
  and completes the induction on $C$.
  The symmetric case follows symmetrically.
  We thus conclude the first half of the lemma.

  Now assume $M = \lambda x.M'$.
  Then we proceed by case analysis on $C$, and observe that it is
  sufficient to show that $N \Ra \lambda x.N'$ and that for all $X
  \vr{\R_1^\star} Y$, either $M'[X/x] \SR_1^\star N'[Y/x]$
  or $M'[X/x] \eccv{\SR_2}{\SR_1^\star} N'[Y/x]$.

  Clearly, $M = \lambda x.M'$ implies $C$ is a canonical context.
  If $C = [\cdot]$, then $\lambda x.M' \R_2 N$ and the claim follows
  immediately from the fact that ${\R} \prog {\SR}$.
  Otherwise, we have $C = \lambda x.C'$, and so $N = \lambda
  x.C'[\wt{N}]$ implies $N \Ra \lambda x.C'[\wt{N}]$ by reflexivity,
  and since ${\R_1^\star} \subseteq {\SR^\star_1} \subseteq
  {\SR^V_2}$, $\lambda x.M' \SR_2^V \lambda x.C'[\wt{N}]$.
  Clearly, $M' = C'[\wt{M}]$ and $N' = C'[\wt{N}]$ and $\fv(M') =
  \fv(N') = \{x\}$.
  But by Lemma \ref{lemma:starsubst}, we get that for all $X
  \vr{\R_1^\star} Y$, $M'[X/x] \R_1^\star N'[Y/x]$ and thus conclude
  $M'[X/x] \SR^V_2 N'[Y/x]$ as desired.
  This exhausts all possible cases for $C$, and since the symmetric
  case follows symmetrically, we conclude the lemma.
\end{proof}
\fi

\begin{lemma}
  \label{lemma:eccvclbv}
  If ${\R} \prog {\SR}$ and ${\R} \subseteq {\SR}$, then $({\R_1},
  {\eccv{\R_2}{\R_1^\star}}) \prog {\SR^V}$.
\end{lemma}


\ifarxivversion
\begin{proof}
  Suppose $M \eccv{\R_2}{\R_1^\star} N$. It is sufficient to show
  that: \begin{inparaenum}[\itshape (i) \upshape]
  \item if $M \ra M'$, then $N \Ra N'$ such that $M' \SR^V_2 N'$, and
    symmetrically for $N$;
  \item if $M = \lambda x.M'$, then there exists an $N'$ such that $N
    \Ra \lambda x.N'$, $\lambda x.M' \SR^V_1 \lambda x.N'$, and for
    all $X \vr{\R_1^\star} Y$, $M'[X/x] \SR^V_2 N'[Y/x]$, and
    symmetrically if $N = \lambda x.N'$.
  \end{inparaenum}

  The hypothesis $M \ra M'$ implies one of the six cases given by
  Lemma \ref{lemma:eccvreduxcases}.
  Let $\alpha$, $\beta$, $\vec{E}_M$, etc., be as in Lemma
  \ref{lemma:eccvreduxcases}, and observe that since ${\R_1} \subseteq
  {\SR_1}$, $E_{Mi} \SR_1 E_{Ni}$ for all $i$.

  In the first case, we're done, for ${\R} \prog {\SR}$.

  In the second and fourth case, if $\alpha \ra \alpha'$, then since
  $\alpha \R_2 \beta$ and ${\R} \prog {\SR}$, $\beta \Ra
  \beta'$ such that $\alpha' \SR^V_2 \beta'$.
  Then this implies $M \ra \vec{E_M}[\alpha']$, $N \Ra
  \vec{E_N}[\beta']$, and by Lemma \ref{lemma:nestrvineccv},
  $\vec{E}[\alpha'] \SR^V_2 \vec{E}[\beta']$.

  In the third case, if $\alpha = \lambda x.\alpha'$, then since
  $\alpha \R_2 \beta$ and ${\R} \prog {\SR}$, $\beta \Ra \lambda
  x.\beta'$ such that for all $X \vr{\R_1^\star} Y$, $\alpha'[X/x]
  \SR^V_2 \beta'[Y/x]$.
  Thus, $M \ra E_{M1}[\cdots[E_{M(n-1)}[\alpha'[V/x]]\cdots]$, $N \Ra
  E_{N1}[\cdots[E_{N(n-1)}[\beta'[V/x]]\cdots]$, and by Lemma
  \ref{lemma:nestrvineccv} we get
  \[E_{M1}[\cdots[E_{M(n-1)}[\alpha'[V/x]]]\cdots] \SR^V_2
  E_{N1}[\cdots[E_{N(n-1)}[\beta'[V/x]]]\cdots]\] as desired.

  In the fifth case, since $\alpha \R_2 \beta$ and $\alpha = \lambda
  x.\alpha'$, $\beta \Ra \lambda x.\beta'' = \beta'$ such that $\alpha
  \vr{\SR_1^\star} \beta'$.
  Let $X \R_1^\star Y$ be such that $E_{Mn} = X[\cdot]$ and $E_{Nn} =
  Y[\cdot]$, then by Lemma \ref{lemma:rstarclbv}, since $X \ra X'$, $Y
  \Ra Y'$ such that $X' \SR^V_2 Y'$.
  Let $E_M' = [\cdot]\alpha$ and $E_N' = [\cdot]\beta'$.
  Then, $M \ra E_{M1}[\cdots[E_{M(n-1)}[E_M'[X']]]\cdots]$ and $N \Ra
  E_{N1}[\cdots[E_{N(n-1)}[E_N'[Y']]]\cdots]$.
  Then by Lemma \ref{lemma:nestrvineccv} we get
  \[ E_{M1}[\cdots[E_{M(n-1)}[E_{M}'[X']]]\cdots] \SR^V_2
  E_{N1}[\cdots[E_{N(n-1)}[E_N'[Y']]]\cdots] \] and we're done.

  Finally, in the sixth case, since $\alpha \R_2 \beta$ and $\alpha =
  \lambda x.\alpha'$, $\beta \Ra \lambda x.\beta'' = \beta'$ such that
  $\alpha \vr{\SR_1^\star} \beta'$.
  Let $X \R_1^\star Y$ be such that $E_{Mn} = X[\cdot]$ and $E_{Nn} =
  Y[\cdot]$, then by Lemma \ref{lemma:rstarclbv}, since $X = \lambda
  x.X'$, $Y \Ra \lambda x.Y'$ such that $X'[\nu/x] \SR^V_2 Y'[\mu/x]$
  for all $\nu \vr{\R_1^\star} \mu$.
  Since ${\R_1^\star} \subseteq {\SR_1^\star}$, we get $X'[\alpha/x]
  \SR^V_2 Y'[\beta'/x]$.
  \RAK{Previous sentence does not follow, and I still haven't figured
    out how to fix it.
    Problem: we're substituting $\SR^\star_1$ into $\R^\star_1$, which
    we aren't necessarily able to do.}
  Thus, $M \ra E_{M1}[\cdots[E_{M(n-1)}[X'[\alpha/x]]]\cdots]$ and $N
  \Ra E_{N1}[\cdots[E_{N(n-1)}[Y'[\beta'/x]]]\cdots]$.
  By Lemma \ref{lemma:nestrvineccv}, we conclude
  \[E_{M1}[\cdots[E_{M(n-1)}[X'[\alpha/x]]]\cdots] \SR^V_2
  E_{N1}[\cdots[E_{N(n-1)}[Y'[\beta'/x]]]\cdots].\]

  The symmetric case follows symmetrically.

  Now assume $M = \lambda x.M'$.
  Clearly, $M = \lambda x.M'$ implies $\vec{E_N} = [\cdot]$.
  Then $\lambda x.M' \R_2 N$ and the claim follows immediately from
  the fact that ${\R} \prog {\SR}$.
  The symmetric case follows symmetrically, and so we conclude the
  lemma.
\end{proof}
\fi

\begin{theorem}
  The up-to-context technique is extensive and respectfully compatible,
  and hence sound.
\end{theorem}

\begin{proof}
  Extensiveness is obvious.
  Assume ${\R} \prog {\SR}$ and ${\R} \subseteq {\SR}$, then by Lemma
  \ref{lemma:rstarclbv} we have $({\R_1}, {\R_1^\star}) \prog
  {\SR^V}$, and by Lemma \ref{lemma:eccvclbv}, $(\R_1,
  {\eccv{\R_2}{\R_1^\star}}) \prog {\SR^V}$.
  Then, by Lemma \ref{prop:clbv-prog-union}, $({\R_1}, {\R_1^\star}
  \cup {\eccv{\R_2}{\R_1^\star}}) \prog {\SR^V}$.
  By the call-by-value analog of Proposition \ref{prop:r1starcompat},
  we deduce that ${\R^V} \prog {\SR^V}$ and we are done.
\end{proof}

\begin{corollary}
  \label{cor:rvclbv}
  If $\R$ is a CLB, then so is $\R^V$.
\end{corollary}

\begin{corollary}
  \label{cor:clbvcontextuallyclosed}
  \leavevmode
  \begin{enumerate}
  \item If $M \clbv_1 N$, then for all contexts $C$, $C[M] \clbv_1
    C[N]$.
  \item If $E \clbv_2 F$, then for all evaluation contexts $\EC$,
    $\EC[E] \clbv_2 \EC[F]$.
  \end{enumerate}
\end{corollary}

\begin{proof}
  In the first case, if $M \clbv_1 N$, then there exists a CLB ${\R}
  \subseteq {\clb}$ such that $M \R_1 N$.
  Then by Corollary \ref{cor:rvclbv}, ${\R^V} \subseteq {\clb}$.
  Since $C[M] \R_1^\star C[N]$ and ${\R^V_1} = {\R_1^\star}$, we get
  $C[M] \clbv_1 C[N]$ as desired.

  In the second case, if $E \clbv_2 F$ then there exists a CLB ${\R}
  \subseteq {\clbv}$ such that $E \R_2 F$. Then again, ${\R^V}
  \subseteq\ {\clbv}$.
  Since ${\eccv{\R_2}{\R_1^\star}} \subseteq {\R^V_2}$ and ${\Id}
  \subseteq {\R_1^\star}$, we have ${\eccn{R_2}{\Id}} \subseteq
  {\R^V_2}$.
  Thus, $E \R_2 F$ implies $\EC[E] \R^V_2 \EC[F]$ for all evaluation
  contexts $\EC$, and since ${R^V_2} \subseteq {\clbn_2}$, we deduce
  the second statement.
\end{proof}

\begin{corollary}
  \label{cor:CLBVinCEV}
  We have the following inclusion of coupled relations: ${\clbv}
  \subseteq (\cev, \ecev)$.
\end{corollary}




\begin{lemma}
  \label{lemma:ecevvalsbeta}
  If $\lambda x.P \ecev \lambda x.Q$, then for all $V \in \V$, $P[V/x]
  \ecev Q[V/x]$.
\end{lemma}

\begin{proof}
  Since ${\Ra} \subseteq {\ecev}$, $(\lambda x.P)V \Ra P[V/x]$ implies
  $(\lambda x.P)V \ecev P[V/x]$ and similarly $(\lambda x.Q)V \ecev
  Q[V/x]$.
  Since $\ecev$ is a congruence relation, $(\lambda x.P)V \ecev
  (\lambda x.Q)V$, and so the lemma follows.
\end{proof}

\begin{theorem}
  \label{thm:cevecevcislbv}
  The coupled relation $(\cev, \ecev)$ is a CLB.
\end{theorem}

\begin{proof}
  Assume $M \ecev N$.
  If $M \ra M'$, then $M \Ra M'$, and so since ${\Ra} \subseteq
  {\ecev}$, $M \ecev M'$.
  Then $N \Ra N$ and by transitivity and symmetry, $M' \ecev N$.

  If $M = \lambda x.P$, then by definition of $\ecev$, $N \Ra \lambda
  x.Q$ for some $Q$, and $\lambda x.P \ecev \lambda x.Q$ by the fact
  that $\ecev$ is an equivalence relation and ${\Ra} \subseteq
  {\ecev}$.
  Then by Theorem \ref{thm:cev-is-ecev}, $\lambda x.P \cev \lambda
  x.Q$ as desired.
  Since ${\cev} \subseteq {\ecev}$ and $({\cev})^\star = {\cev}$, by
  Lemma \ref{lemma:starsubst}, we get that $P[V/x] \ecev Q[W/x]$ for
  all $V \vr{({\cev})^\star} W$, as desired.
\arxiv{

 }{}
  Since the symmetric cases follow symmetrically, we derive the
  theorem.
\end{proof}

\begin{corollary}
  \label{cor:cevecevequalsclbv}~

  Coupled logical bisimilarity coincides with the contextual
  equivalences, i.e., $(\clbv_1, \clbv_2) = (\cev, \ecev)$.

\arxiv{
\end{corollary}

\begin{proof}
  By double inclusion via Corollary \ref{cor:CLBVinCEV} and Theorem
  \ref{thm:cevecevcislbv}.
\end{proof}

\begin{corollary}
}{}
  Coupled logical bisimilarity, $\clbv$, is a CLB.
\end{corollary}

\begin{proof}
  Immediate by Theorems \ref{thm:cevecevcislbv} and
  \ref{thm:cev-is-ecev}.
\end{proof}

\subsection{Applicative Bisimulation}

The call-by-value version of applicative bisimulation is nearly
identical to the call-by-name version, apart from the obvious
restriction to values in the substitution clause:
\begin{definition}
  A relation $\R\ \subseteq \Ld\times\Ld$ is called an
  \textbf{applicative bisimulation} if $M\R N$ implies whenever $M \Ra
  \lambda x.P$, $N \Ra \lambda x.Q$ for some $Q$ and $P[W/x] \R
  Q[W/x]$ for all $W \in \V$, and conversely for $N$.
  We call the union of all applicative bisimulations, $\ABv$, Ass
  \textbf{applicative bisimilarity}.
\end{definition}

As in the call-by-name case, applicative bisimilarity is an
applicative bisimulation and we seek to show that applicative
bisimilarity coincides with contextual equivalence, i.e., is a
congruence.
In contrast to the call-by-name case, it is not easy to give a direct
proof that applicative bisimilarity can be seen as a CLB for the
call-by-value $\lambda$-calculus.
Assume we tried the naive approach we used in the call-by-name case,
and claimed that the coupled relation $(\vr{\Id}, \ABv)$ is a CLB.
We show that this claim is false: one can easily show that $\lambda
x.(\lambda y.y) x \ABv \lambda x.x$.
However, clause 2 of the definition of CLB then requires that $\lambda
x.(\lambda y.y) x \mathrel{\Id} \lambda x.x$, which is clearly false.
Thus, the proposed embedding of applicative bisimilarity into a CLB is
incorrect.
We could correct this deficiency by proposing instead the embedding
$(\vr{\ABv}, \ABv)$.
However, this embedding is no longer faithful to the spirit of
applicative bisimulation, since we now permit the substitution of
non-identical pairs into values related by $\ABv$.
This problem motivates the need for an analogue of the \textit{up-to
  environment} proposed in \cite{Sangiorgi2007} for logical
bisimulation:\RAK{Check that this reference is correct.}

\RAK{Find a way to define up-to environment using progression seeing
  that it's one of our main assets}
\begin{definition}
  A coupled relation $\R$ is said to be a \textbf{CLB up-to
    environment} if it satisfies all clauses of the definition of CLB
  except for the requirement that $\lambda x.M' \R_1 \lambda x.N'$ in
  clause 2.
\end{definition}

\subsection{Logical Bisimulation}

Although we cannot faithfully embed applicative bisimulation into
call-by-value CLBs, we can still embed logical bisimulations (LBs). We
present the call-by-value version of \arxiv{the logical
  bisimulation}{LB} as 
introduced by \textcite{Sangiorgi2007}.

\begin{definition}
  A relation $\R\ \subseteq\ \Ld\times\Ld$ is called a \textbf{logical
    bisimulation} if whenever $M \R N$:
  \begin{enumerate}
  \item if $M \ra M'$, then $N \Ra N'$ and $M' \R N'$;
  \item if $M = \lambda x.M'$, then $N \Ra \lambda x.N'$, and for all
    $V \vr{\R^\star} W$, $M'[V/x] \R N'[W/x]$;
  \item the converses for $N$.
  \end{enumerate}
  The union of all logical bisimulations is called \textbf{logical
    bisimilarity} and is denoted $\lbv$.
\end{definition}

Although it is claimed in \cite{Sangiorgi2007} that to have soundness,
we must additionally require that $\lambda x.M' \R \lambda x.N'$ in
clause 2 of the definition, that is to say, that clause 2 should read
as ``if $M = \lambda x.M'$, then $N \Ra \lambda x.N'$, and for all $V
\vr{\R^\star} W$, $M'[V/x] \R N'[W/x]$'', the following proposition
shows that this
additional requirement is redundant:

\begin{prop}
  \label{prop:R-relates-values}
  If $\R$ is a relation such that whenever $M \R N$:
  \begin{enumerate}
  \item if $M \ra M'$, $N \Ra N'$ and $M' \R N'$;
  \item if $M = \lambda x.M'$, then $N \conv$; and
  \item the converses of the previous two conditions for $N$;
  \end{enumerate}
  then whenever $\lambda x.M \R N$, we have that $N \conv \lambda x.N'$ and
  $\lambda x.M \R \lambda x.N'$.
\end{prop}

\begin{proof}
  Suppose $P \R N$. By induction on $N\Ra Q$, using 1, we get $P \Ra
  P'$ such that $P' \R Q$.  With $P = \lambda x.M$ and $Q' = \lambda
  x.N'$ for some $N'$ (thanks to 2) we must have $P' = \lambda x.M$ as
  well.
\end{proof}


Unfortunately, we have not been able to similarly drop the requirement
that $\lambda x.M' \R_1 \lambda x.N'$ in the second clause of the
definition of call-by-value CLBs.

As one would expect, we still have the following proposition, due to
\cite[Corollary 1, Lemma 4, and p. 12]{Sangiorgi2007}:

\begin{prop}[\cite{Sangiorgi2007}]
  \label{prop:lbvisbisimandcong}
  Logical bisimilarity is the largest logical bisimulation and is a
  congruence\arxiv{ relation.}{.}
\end{prop}

Moreover, as in the call-by-name case, and thanks to Proposition
\ref{prop:R-relates-values}, we have that:

\begin{prop}
  \label{prop:lbiffclbv}
  A relation $\R$ is a logical bisimulation if and only if $(\R, \R)$
  is a CLB.
\end{prop}

\begin{corollary}
  Logical bisimilarity, coupled logical bisimilarities, and contextual
  equivalences coincide, i.e., ${\lbv} = {\clbv_2} = {\ecev} = {\cev}
  = {\clbv_1}$.
\end{corollary}

\section{Concluding Remarks}

Logical bisimulations build upon applicative bisimulations and make
proofs of congruence simpler, without relying on Howe's
method~\cite{Howe1996,Pitts2011}.
Up-to techniques for logical bisimulations can be
defined~\cite{Sangiorgi2007},  in order to bisimulation proofs easier.
Their definition, however, is rather ad hoc.

Coupled logicial bisimulations bridge the gap between applicative and
logical bisimulations: indeed, the latter are special cases of CLBs.
One can reach applicative bisimulation by making the first component
of a CLB as small as possible.  For it to correspond to logical
bisimulation, the first component needs to be larger, to the point of
being equal to the second component.

We need to study further coupled logical bisimulations, in order in
particular to draw a comparison with environmental bisimulation. While
in the latter, intuitively, we need to make environments grow along
the development of an equivalence proof, an interesting feature of CLB
is the possibility to keep the first component of a coupled relation small.

Possible extensions of this work include treating richer
$\lambda$-calculi, like a $\lambda$-calculus with imperative features,
for which a notion of state have to be introduced.


\paragraph{Acknowledgements}

Discussions with Davide Sangiorgi, Damien Pous and Daniel Hirschkoff
have been helpful in the development of this work.
The authors acknowledge the support of the 
ANR 12IS02001 PACE project.

\bibliographystyle{abbrv}
\bibliography{ref}
\end{document}